\documentclass[preprint,10pt,fleqn,authoryear]{elsarticle}

\usepackage[section]{placeins}
\usepackage{amsfonts}
\usepackage{amssymb}
\usepackage[top=2cm,bottom=2.5cm,right=1.5cm,left=1.5cm]{geometry}
\usepackage{amsthm}
\usepackage{amsmath}
\usepackage{graphicx}
\usepackage{caption}
\usepackage[colorlinks=true, allcolors=blue]{hyperref}
\usepackage{bm}
\newcommand{\arccot}{\mathrm{arccot}\,}

\def \d{{\textrm d}}
\newtheorem{thm}{Theorem}
\newtheorem{lemma}{Lemma}

\begin{document}

\captionsetup[figure]{labelfont={bf},labelformat={default},labelsep=period,name={Fig. }}

\begin{frontmatter}



\title{A generalized geometric mechanics theory for multi-curve-fold origami: vertex constrained universal configurations}


\author[inst1]{Zhixuan Wen}

\affiliation[inst1]{organization={State Key Laboratory for Turbulence and Complex Systems, Department of Mechanics and Engineering Science, BIC-ESAT, College of Engineering, Peking University},
            city={Beijing},
            postcode={100087}, 
            country={China}}

\author[inst1]{Pengyu Lv}
\author[inst3]{Fan Feng \corref{cor1}}
\author[inst1,inst2]{Huiling Duan \corref{cor1}}

\affiliation[inst2]{organization={HEDPS, CAPT and IFSA, Collaborative Innovation Center of MoE, Peking University},
            city={Beijing},
            postcode={100087}, 
            country={China}}

\affiliation[inst3]{organization={Department of Mechanics and Engineering Science, College of Engineering, Peking University},
            city={Beijing},
            postcode={100087}, 
            country={China}}
\cortext[cor1]{\noindent Corresponding authors \newline  \indent E-mail addresses: fanfeng@pku.edu.cn (Fan Feng), hlduan@pku.edu.cn (Huiling Duan)}  
    
\begin{abstract}
  Folding paper along curves leads to spatial structures that have curved surfaces meeting at spatial creases, defined as curve-fold origami. In this work, we provide an Eulerian framework focusing on the mechanics of arbitrary curve-fold origami, especially for multi-curve-fold origami with vertices. We start with single-curve-fold origami that has wide panels. Wide panel leads to different domains of mechanical responses induced by various generator distributions of the curved surface. The theories are then extended to multi-curve-fold origami, involving additional geometric correlations between creases. 
  As an illustrative example, the deformation and equilibrium configuration of origami with annular creases are studied both theoretically and numerically. 
  Afterward, single-vertex curved origami theory is studied as a special type of multi-curve-fold origami. We find that the extra periodicity at the vertex strongly constrains the configuration space, leading to a region near the vertex that has a striking universal equilibrium configuration regardless of the mechanical properties. Both theories and numerics confirm the existence of the universality in the near-field region. In addition, the far-field deformation is obtained via energy minimization and validated by finite element analysis. Our generalized multi-curve-fold origami theory, including the vertex-contained universality, is anticipated to provide a new understanding and framework for the shape programming of the curve-fold origami system.
\end{abstract}


\end{frontmatter}

\tableofcontents
\section{Introduction}
\label{sec1}
Origami is an ancient art form aiming to achieve complex spatial shapes by folding planar sheets. Classic origami only involves folding along straight lines, and it was not until the 1920s that students in Bauhaus realized sheets could be folded along curves, unveiling a novel branch that expands the traditional art form \citep{CurvedCrease_AAG2008, demaine2011curved}. This new branch, defined as curve-fold origami, opens new possibilities in origami design. Beyond art design, recently curve-fold origami has found applications in soft robotics \citep{doi:10.1126/scirobotics.aaz6262,doi:10.1126/scirobotics.aat0938,feng2024geometry}, metamaterials \citep{lee2021compliant, PRLmeta,meta3}, architecture \citep{tachi2011designing,mouthuy2012overcurvature} and virtual reality technology \citep{VR}, due to its extraordinary abilities in energy storage, fast deployment and stiffness manipulation.

The scientific exploration of curve-fold origami originated with the pioneering work by David \citet{1674542}. In the groundwork, curved origami structures are assumed to deform isometrically, thus they are geometrically modeled as developable surfaces meeting at space curves. Harnessing the theory of curves and surfaces in differential geometry \citep{do1976differential}, the geometrical theories of curved origami are developed in subsequent works \citep{duncan1982folded,10.2307/2589583,kilian2008curved,demaine2015characterization}, revealing the geometric correlations between curved folds and bent panels. The foundational results in this area are:
\emph{Given the reference strips, the two developable surfaces are determined by the space crease in the deformed configurations.}
The theories contribute to the computational design of developable surfaces \citep{10.1145/3180494} and curve-fold origami \citep{mitani2011design,mitani2011interactive,tachi2013composite,jiang2019curve,sasaki2022simple,mundilovacurved}, as well as the predicted shape of M\"{o}bius strip \citep{starostin2007shape,audoly2023analysis}.

Built upon geometric correlations, the mechanical properties of single-curve-fold origami have been investigated.  The isometric deformation of curve-fold origami is determined via elastic energy (creases' folding energy + panels' bending energy) minimization under the constraints of developability. Mechanically, creases are modeled as elastic hinges distributed along curves and panels are modeled as unstrechable Kirchhoff plates with energy density proportional to the principal curvatures' square.  Due to the foundational results in geometric works, the bending energy of developable surfaces reduces to a $1$-dimensional integral of the energy functional associated with the geometry of reference curves. The exact form of the bending energy is derived by Wunderlich \citep{Wunderlich1962,dias2015wunderlich,todres2015translation}, which degenerates to the more widely used \citet{sadowsky1930theorie} functional in narrow panels \citep{starostin2007shape,dias2012geometric,DIAS201457,YU2019657}. 
Taken together, the entire elastic energy is expressed as an integral along the crease, determined by the crease's geometry. Minimizing the energy yields the equilibrium geometry of the deformed crease as well as the configuration of the entire origami. 

In the area of geometric mechanics, the deformation of an annular sheet folding along its centerline is studied in \citet{dias2012geometric}. The energy functional is expressed with curvature $\kappa$ and torsion $\tau$ of the deformed fold. Minimizing the energy leads to an equilibrium configuration that can be described analytically with $\kappa,\tau$. The theoretical results capture the buckling phenomenon well. The theory is further extended to derive a nonlinear rod model for folded elastic narrow strips \citep{DIAS201457}. In the work, curved origami is fitted into the framework of thin rods \citep{MOULTON2013398}, which simplifies numerical solving and stability analyses. In other studies, panels are constrained to deform cylindrically, thus the deformation is derived by generalized Euler-Bernoulli beam theory \citep{lee2018elastica}. However, a general theory for curved origami with arbitrary deformation and reference shapes is absent. Beyond single-curve-fold origami, some previous works have studied multiple curved folds, including the geometric design \citep{dias2012thesis, liu_design_2024}, and the energy of folds with constant curvature and torsion (for example, a comprehensive study for helicoids) \citep{Dias_2012}. However, the Euler-Lagrange equations (or equilibrium equations) for multi-curve-fold systems remain unknown. Thus, how geometry and elasticity contribute to the deformation is still poorly understood. Furthermore, the theory for multi-curve-fold origami with vertices is still lacking.

In this paper, we derive a generalized theory for the deformation of multi-curve-fold origami, particularly with vertices. Our theory reveals a remarkable universality in this system: the deformed configuration of a multi-curve-fold origami near the vertex is solely determined by the symmetry at the vertex and independent of the mechanical properties. To this end, we first derive a generalized theory for 
single-curve-fold origami, removing the narrow panel assumption in many related works. The most significant differences between narrow and wide panels are: geometrically, all generators of a developable surface start from the crease and end at the opposite edge, while in wide panels, generator distribution is more complex; mechanically,  different generator distributions yield various mechanical responses. On the geometry side, we classify the generator distribution comprehensively for wide panels, and on the mechanics side, we utilize the Wunderlich theory to derive the energy density for different generator distributions precisely. Equilibrium configurations are then obtained via energy minimization.

Based on the single-curve-fold theory,
we extend the theoretical framework to multi-curve-fold origami by introducing geometric correlations between creases. Such correlations result in the propagation of generators and deformations between adjacent creases, which are solved by a numerical scheme. Accordingly, 
the energy distribution is derived, leading to the geometric mechanics theory of multi-curve-fold origami. Beyond existing works, we derive the Euler-Lagrange equations revealing how geometry and elasticity determine the multi-curve-fold systems' deformation. Validations are made on an illustrative structure with multiple annular folds, where theoretical and finite element analysis (FEA) results fit well.

More interestingly, we find that a mechanics-independent deformed configuration emerges when the multiple creases meet at a vertex. We prove that in the system of multi-curve-fold origami with a vertex, the periodicity at the vertex yields a universal equilibrium configuration in the domain near the vertex,
regardless of mechanical properties (stiffness of folds and panels) and geodesic curvatures of folds. Specifically, the deformed configuration is determined only by the folding angle at the vertex. In contrast, 
in the far-field domain determined by the aforementioned geometric correlations, the energy minimization yields various equilibrium configurations for different mechanical constants. Both results are validated by FEA. The theories pioneerly indicate how multi-curve-fold origami with vertices deform.

The paper is organized as follows. In Section \ref{sec2}, we start with the preliminaries of differential geometry for the deformation of a single strip and then develop a generalized theory for single-curve-fold origami with arbitrary reference shapes and wide panels. The theory is extended to multi-curve-fold theory by deriving the correlations between creases in Section \ref{sec2new}. Based on the theory, multi-curve-fold origami with a vertex is theoretically studied in Section \ref{sec3}\label{mathrefs}, which is compared with numerical results. Finally, Section \ref{sec:conclusion} concludes the main points of this paper.

\section{Single-curve-fold origami theory}\label{sec2}
\subsection{Preliminaries: geometry of single-curve-fold origami} \label{sec2.1}
We start with the classic differential geometry for the shape description of curve-fold origami in the deformed domain, i.e., in the Eulerian framework. In this paper, we assume the origami structure is deformed isometrically from a flat sheet. Therefore, the flanks are modeled as developable surfaces, which are generically given by
\begin{equation}\label{e1}
\mathbf{r}(S,v)=\mathbf{r}_0(S)+v \mathbf{l}(S),
\end{equation}
with the developability constraint
\begin{equation} \label{e2}
\mathbf{r}_0^{\prime} \cdot (\mathbf{l} \times \mathbf{l}^{\prime})=0.
\end{equation}
Here $\mathbf{r}_0(S)$ is a selected reference curve on the surface, and $\mathbf{l}(S)$ is the generator at the arclength $S$. In curve-fold origami, it is convenient to choose the curved crease as the reference curve and the flanks on the two sides are described by Eq. \eqref{e1}. Solving these two equations yields strong geometric correlations between the reference curve and the developable surface, which are given explicitly via solving the kinematic equations of Frenet frames and Darboux frames (shown in Fig. \ref{fig:sec2122}(a)).

Frenet frame is a moving frame on a space curve that consists of the unit tangent $\mathbf{t}$, the unit normal $\mathbf{n}$, and the unit binormal $\mathbf{b}=\mathbf{t} \times \mathbf{n}$. The derivatives of the vectors along the curve are
\begin{equation}\label{e3}
\begin{aligned}
     \mathbf{t}^{\prime}&=\kappa \mathbf{n},\\  
     \mathbf{n}^{\prime}&=-\kappa \mathbf{t} + \tau \mathbf{b},\\ 
     \mathbf{b}^{\prime}&=-\tau \mathbf{n},
\end{aligned}
\end{equation}
where $\kappa$ and $\tau$ are the curvature and torsion, which determine the shape of the space curve uniquely up to rigid motions by the fundamental theorem of curves \citep{do1976differential}.

Darboux frame is a moving frame on a surface that consists of $\mathbf{e}_1$, $\mathbf{e}_2$, $\mathbf{e}_3$ where $\mathbf{e}_1$, $\mathbf{e}_2$ are on the tangent plane and $\mathbf{e}_3$ is the unit normal of the surface. For a developable surface, we select
\begin{equation}\label{e4}
\begin{aligned}
     \mathbf{e}_1&=\mathbf{t},\\  
     \mathbf{e}_3&=\frac{\mathbf{e}_1 \times \mathbf{l}}{\left|\mathbf{e}_1 \times \mathbf{l}\right|},\\ 
     \mathbf{e}_2&=\mathbf{e}_3 \times \mathbf{e}_1,
\end{aligned}
\end{equation}
to connect the different frames, and the selection makes sure that the frame will not change its direction when moving along a generator. Thus the Darboux vectors only relate to the arc length $S$ and the kinematic equations are
\begin{equation}\label{e5}
\begin{aligned}
     \d\mathbf{r} &= \omega_1 \mathbf{e}_1+ \omega_2 \mathbf{e}_2,\\
     \d\mathbf{e}_1&=\omega_{12} \mathbf{e}_2 + \omega_{13} \mathbf{e}_3, \\  
     \d\mathbf{e}_2&=\omega_{21} \mathbf{e}_1 + \omega_{23} \mathbf{e}_3, \\ 
     \d\mathbf{e}_3&=\omega_{31} \mathbf{e}_1 + \omega_{32} \mathbf{e}_2, \\
     \omega_{ij}&=-\omega_{ji},
\end{aligned}
\end{equation}
where $\omega_{ij}$ is the differential form representing the rotation of the Darboux frame.
According to classical differential geometry (page 63 of \citep{do1976differential}), the nonzero principal curvature $\kappa_m$ satisfies
\begin{equation}\label{e6}
\begin{aligned}
&[\omega_{13}, \omega_{23}]=[\omega_1, \omega_2] B,\\
&\kappa_m = \mathrm{Tr}(B),
\end{aligned}
\end{equation}
where $[\omega_{13},\omega_{23}]$ and $[\omega_1, \omega_2]$ are $1\times 2$ matrices and $B$ is a $2\times 2$ symmetric matrix related to the second fundamental form of the surface. Geodesic curvature $\kappa_g$ and normal curvature $\kappa_n$ of the curve on surface are defined as
\begin{equation}\label{e7}
\begin{aligned}
    \kappa_g&=\kappa \mathbf{n} \cdot \mathbf{e}_2, \\
    \kappa_n&=\kappa \mathbf{n} \cdot \mathbf{e}_3. 
\end{aligned}
\end{equation}
The geodesic curvature is constant during isometric deformation, thus it only depends on the reference planar shape. To describe the distribution of generators we define
\begin{equation} \label{e8}
t=\cot \left\langle\mathbf{e}_1, \mathbf{l}\right\rangle,
\end{equation}
\begin{figure}[t]
    \centering
    \includegraphics[width=1\textwidth]{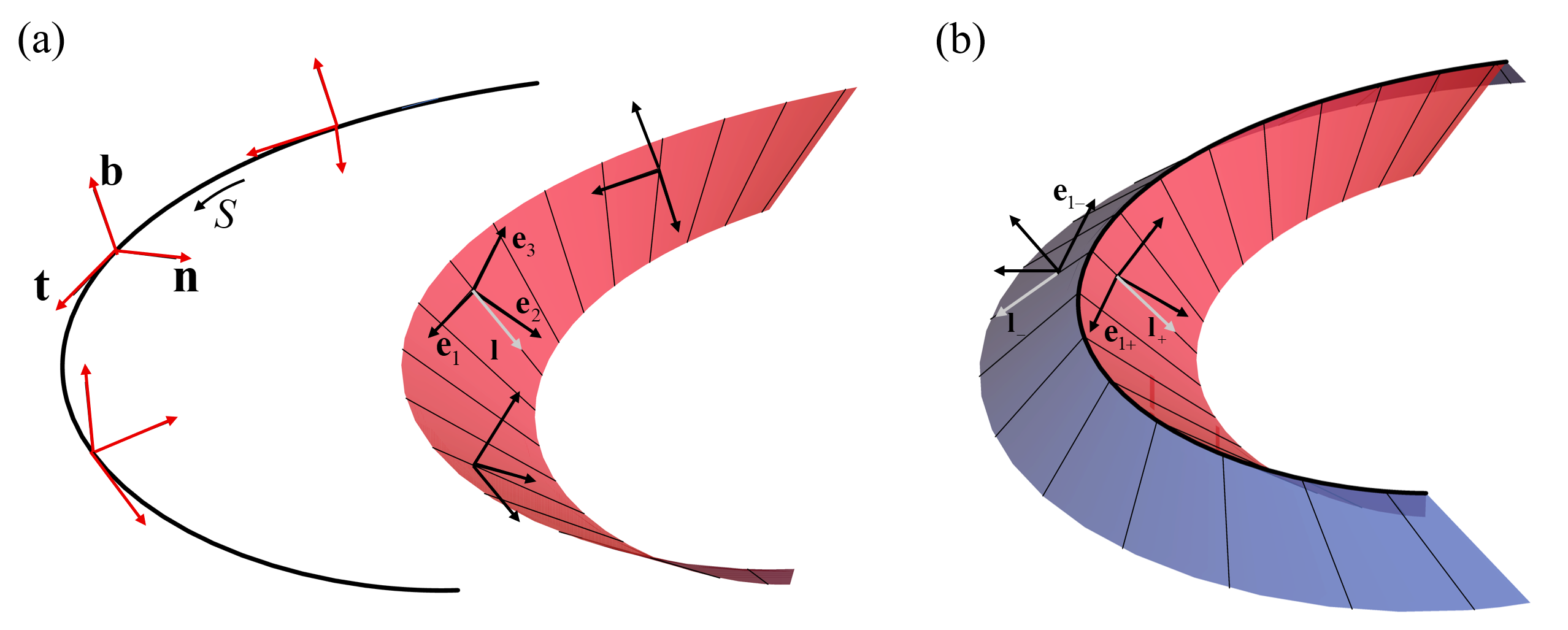}
    \caption{\label{fig:sec2122}Illustrations of Frenet frames and Darboux frames: (a) Frenet frames on a reference curve and Darboux frames on a surface; (b) Darboux frames on panel $+$ (red) and panel $-$ (blue).}
\end{figure}
where $\left\langle\mathbf{e}_1, \mathbf{l}\right\rangle$ represents the angle between the vectors $\mathbf{e}_1$ and $\mathbf{l}$, ranging from $0$ to $\pi$. The generator vector $\mathbf{l}$ is selected as
\begin{equation} \label{e9}
\mathbf{l}=t \mathbf{e}_1+\mathbf{e}_2.
\end{equation}
According to Eq. \eqref{e2} - \eqref{e7}, the derivatives are written as
\begin{equation}\label{e10}
\begin{aligned}
     \mathbf{e}_1^{\prime}&=\kappa_{g} \mathbf{e}_2 + \kappa_{n} \mathbf{e}_3, \\  
     \mathbf{e}_2^{\prime}&=-\kappa_{g} \mathbf{e}_1 - t \kappa_{n} \mathbf{e}_3, \\ 
     \mathbf{e}_3^{\prime}&=-\kappa_{n} \mathbf{e}_1 + t \kappa_{n} \mathbf{e}_2, \\
\end{aligned}
\end{equation}
and the nonzero principal curvature is derived by
\begin{equation} \label{e11}
\kappa_m=\frac{\left(1+t^2\right) \kappa_n}{1+v t^{\prime}-v \kappa_g-v t^2 \kappa_g}.
\end{equation}
On the premise of developability, the generators cannot intersect with each other inside the surface, otherwise, the intersection will introduce divergent principal curvature and divergent bending energy discussed later. To see this, let $\mathbf{r}(S^\star, v^\star)$ be the generator intersection, satisfying $(\partial \mathbf{r}(S^\star, v^\star)/ \partial S) \cdot \mathbf{e}_1=0$. By direct calculation, we have $(\partial \mathbf{r}(S^\star, v^\star)/ \partial S) \cdot \mathbf{e}_1=1+v^\star t^{\prime}-v^\star \kappa_g-v^\star t^2 \kappa_g =0$, resulting in divergent principal curvature at $(S^\star, v^\star)$ according to Eq. (\ref{e11}).

Furthermore, the torsion can be expressed using Darboux frame
\begin{equation}\label{e12}
    \tau=\frac{\mathbf{t} \cdot (\mathbf{t}^{\prime}\times \mathbf{t}^{\prime \prime})}{\kappa^2}=\frac{\mathbf{e}_1 \cdot (\mathbf{e}_1^{\prime}\times \mathbf{e}_1^{\prime \prime})}{\kappa^2}=-t \kappa_n+\frac{\kappa_g}{\kappa^2}\kappa_n^{\prime}-\frac{\kappa_n}{\kappa^2}\kappa_g^{\prime}.
\end{equation}
Therefore, the principal curvature $\kappa_m(S,v)$ distribution on the panel depends on the curvature $\kappa(S)$ and torsion $\tau(S)$ of the curved fold, according to Eq. \eqref{e7}, \eqref{e11} and \eqref{e12}. Symmetrically we derive the geometric variables in the opposite panel
\begin{equation}\label{e13}
\begin{aligned}
    \kappa_{g\pm}&=\kappa \mathbf{n} \cdot \mathbf{e}_{2\pm}, \\
    \kappa_{n\pm}&=\kappa \mathbf{n} \cdot \mathbf{e}_{3\pm}, \\
    t_{\pm}&=\cot \langle\mathbf{e}_{1\pm}, \mathbf{l}_{\pm}\rangle,\\
    \kappa_{m_\pm}&=\frac{(1+t_\pm^2) \kappa_{n\pm}}{1 \pm v_{\pm} t_{\pm}^{\prime}-v_{\pm} \kappa_{g\pm}-v_{\pm} t_{\pm}^2 \kappa_{g\pm}},\\
    \tau&=-\mathrm{t}_{\pm} \kappa_{n \pm} \pm \frac{\kappa_{g \pm}}{\kappa^2} \kappa_{n \pm}^{\prime} \mp \frac{\kappa_{n \pm}}{\kappa^2} \kappa_{g \pm}^{\prime}.
\end{aligned}
\end{equation}
Notice that the Darboux frames of panel $-$ and panel $+$ are illustrated in Fig. \ref{fig:sec2122}(b), where the positive directions of ${\bf e}_1$ are defined in opposite directions for convenience. For single-curve-fold origami with a nonzero $\kappa_{g+}+\kappa_{g-}$ is classified as non-developable curved origami \citep{Fannon,feng2022interfacial,zou2024kinematics}, which can be obtained by stitching together two separate curved strips. In this paper, we focus on the developable case ($\kappa_{g+}+\kappa_{g-}=0$), meaning that the origami is folded from a single sheet of paper.

Folding angle is a key characteristic of origami involving folding energy. In curved origami, the folding angle may vary along the curve, defined as
\begin{equation}\label{e14}
\phi=\langle\mathbf{e}_{3+},\mathbf{e}_{3-}\rangle=|\arccos(\frac{\kappa_{n+}}{\kappa}) \pm \arccos(\frac{\kappa_{n-}}{\kappa})|,
\end{equation}
which only depends on $\kappa$.

\subsection{Domain classification for curve-fold origami with wide flanks}\label{class}

As analyzed in the previous section, generators determine the deformation of panels. Therefore, before moving to mechanics, a discussion on domain classification based on generator distribution is needed. Previous works concerning curved origami assume that all generators starting from the crease end at the opposite edge, while in wide panels this assumption breaks. A comprehensive domain classification is analyzed in this section.

Some works have discussed the generator distribution in a single developable surface \citep{solomon,CHEN2022105068}. As for origami panels, the geometric constraints between folds and panels induce a more complex generator distribution. The generators are classified based on their starting and ending points, and panels are divided into different segments as shown in Fig. \ref{fig:sec222}(a).  Such classification is essential for elastic energy analysis in the following sections. Focusing on smooth developable folds, the properties and distinctions of different domains are analyzed as follows.   

\textbf{Crease-edge domain}. In this domain, generators start at creases and end at edges. In previous works concerning a narrow strip, the crease-edge domain covers the whole panel.  

\textbf{Edge-edge domain}. In this domain, generators start and end at edges. When no external loads are applied on the panel, the edge-edge domain remains planar to satisfy the balance of moments.

\textbf{Plane domain}.  If the generator distribution function $t(S)$ jumps, an infinite $t^{\prime}(S)$ occurs on crease. According to Eq. \eqref{e11}, it leads to $\kappa_m(S,v)=0$. Geometrically, it represents the situation where two generators intersect at one point on the crease, forming a plane domain. We then discuss possible situations that cause the jump in $t(S)$. 
We assume $\mathbf{e_{i\pm}}$ remains continuous over each panel and $\kappa_n$ remains finite along the fold. Otherwise, the discontinuity will induce an infinite curvature and unphysical divergent bending energy. For developable folds, Eq. \eqref{e13} is converted into
\begin{equation}\label{etclass}
t_{\pm}=\frac{1}{\kappa_{n+}}(\frac{\kappa_{g+}}{\kappa^2}\kappa_{n+}^{\prime}-\frac{\kappa_{n+}}{\kappa^2}\kappa_{g+}^{\prime}\pm\tau).
\end{equation}

Analysing Eq. \eqref{etclass}, we give all possible situations that form jumps in $t_{\pm}$ and plane domains.
\begin{enumerate}
\item \noindent$\kappa_{n+}^{\prime}\rightarrow\infty$  and $\kappa_{g+}^{\prime}$ remains finite. From Eq. \eqref{etclass}, since $t_+$ and $\kappa_{n+}$ remain finite, the infinite $\kappa_{n+}^{\prime}$ leads to an infinite $\tau$, which results in a discontinuity of Frenet normal $\mathbf{n}$. Notice that from the geodesic curvature constraint we derive
\begin{equation}\label{gcons}
\frac{\mathbf{n}\cdot\mathbf{e_{2+}}}{\mathbf{n}\cdot\mathbf{e_{2-}}}=\frac{\kappa_{g+}}{\kappa_{g-}}=-1,
\end{equation}
then a discontinuity in $\mathbf{n}$ breaks this constraint, since $\mathbf{e}_{2\pm}$ remain continuous. Therefore, such a situation cannot occur.

\item \noindent $\kappa_{n+}^{\prime}\rightarrow\infty$ and $\kappa_{g+}^{\prime}\rightarrow\infty$. In developable folds, $\kappa_{g+}/\kappa_{g-}$ is constant in jumps of $\kappa_g$. Therefore, from Eq. \eqref{gcons}, in this case the jumps of $\kappa_{n+}$ and $\kappa_{g+}$ should keep the continuity of $\mathbf{n}$, i.e., a finite $\tau$. With Eq. \eqref{etclass}, $\kappa_{n+}^{\prime}/\kappa_{g+}^{\prime}=\kappa_{n+}/\kappa_{g+}$ should be satisfied to derive a finite $\tau$ and form a plane domain.

\item \noindent $\tau^{\prime}\rightarrow\infty.$ In this case, $t$ jumps and a plane domain forms. $\kappa_{n+}$ and $\kappa_{g+}$ may jump as well, and to avoid panel intersections, the jump of $t_{\pm}$ at the jumping point  satisfies $\Delta t_{+}\geq 0$ and $\Delta t_{-} \leq 0$, constraining $\Delta \tau$, $\Delta \kappa_{n+}^{\prime}$ and 
$\Delta \kappa_{g+}^{\prime}$.

\end{enumerate}

In summary, a plane-plane domain forms in the following two situations: (1) $\kappa_{n+}^{\prime}\rightarrow\infty$ and $\kappa_{g+}^{\prime}\rightarrow\infty$ with $\kappa_{n+}^{\prime}/\kappa_{g+}^{\prime}=\kappa_{n+}/\kappa_{g+}$; (2) $\tau^{\prime}\rightarrow\infty$.

\textbf{Crease-crease domain}. In this domain, the generators start from one crease and end at the same crease. At a limiting point $S_{c}^{*}$, the length of the generator tends to zero, where $\kappa_{g+}>0$ and $t_+$ jumps rapidly from $+\infty$ to $-\infty$. In the following discussion we discuss how generator distribute in panel $-$. To derive the generator distribution function $t_-(S)$, we focus on an infinitesimal region around $(S^*_c)$ and define $\kappa^*_{n}(\xi)=\kappa_{n+}(S^*_c+\xi)$, $\kappa^*_{g}(\xi)=\kappa_{g+}(S^*_c+\xi)$, $\kappa^*(\xi)=\kappa(S^*_c+\xi)$, $\tau^*(\xi)=\tau(S^*_c+\xi)$ and $t^*_{\pm}(\xi)=t_{\pm}(S^*_c+\xi)$ with $\xi \rightarrow 0$. Due to possible singularities in Eq. \eqref{etclass}, $t_-(S)$ is given by analyzing the order of the parts $\kappa_{g}^*\kappa_n^{*\prime}/\kappa^{*2}\kappa_n^*$, $\kappa_g^{*\prime}/\kappa^{*2}$ and $\tau^*/\kappa_n^*$ in Eq. \eqref{etclass}.

We start with the order of $t_+$, which is derived from
\begin{equation}\label{ecct}
    t^*_+(\xi)=O(\cot(-\xi\kappa^*_g(\xi)))=O(-(\xi\kappa^*_g(\xi))^{-1}).
\end{equation}

Considering smooth folds, $\kappa_g^*$ is finite. Therefore, $\kappa^{*\prime}_g(\xi)=o(\kappa_g^*\xi^{-1})$. Thus we derive
\begin{equation}\label{ecckg}
    \frac{\kappa_g^{*\prime}(\xi)}{\kappa^*(\xi)^2}=o((\xi\kappa^*_g(\xi))^{-1})=o(-t^*_+(\xi))
\end{equation}

To keep the principal curvature $\kappa_m$ finite, an infinite $t$ derived from Eq. \eqref{ecct} should keep $(1+t^{*2}(\xi))\kappa^*_n(\xi)$ finite. Therefore, $\kappa^*_n(\xi)$ is at most $O(\xi^2\kappa^*_{g}(\xi))$. We then derive
\begin{equation} \label{ecckn}
    \frac{\kappa_g^*(\xi)\kappa_n^{*\prime}(\xi)}{\kappa^{*2}(\xi)\kappa_n^{*}(\xi)}=O((\xi\kappa^*_g(\xi))^{-1}).
\end{equation}

From Eq. \eqref{etclass}, \eqref{ecct}, \eqref{ecckg} and \eqref{ecckn}, the order of $\tau^*/\kappa^*_n$ is derived
\begin{equation}
    \frac{\tau^*(\xi)}{\kappa_n^*(\xi)}=O((\xi\kappa^*_g(\xi))^{-1}),
\end{equation}
which leads to
\begin{equation}
    t^*_{-}(\xi)=O((\xi\kappa^*_g(\xi))^{-1})=O(-t^*_{+}(\xi)).
\end{equation}
Therefore, $t-$ in panel $-$ jumps from $-\infty$ to $+\infty$ rapidly. To avoid the intersection of generators in panel $-$, panel $-$ must be divided into two pieces at $S_c^*$. To help understand this conclusion we give a simple example of curved origami with crease-crease domain. As demonstrated in Fig.~\ref{fig:sec222}(b), a semi-circle fold is folded into
\begin{equation} \label{e16}
    \mathbf{r}_c(S)=(\cos(\frac{S}{R}),\frac{2R}{\pi}-\frac{2R}{\pi} \cos(\frac{\pi}{2}\sin(\frac{S}{R})),\frac{2R}{\pi} \sin(\frac{\pi}{2}\sin(\frac{S}{R}))).
\end{equation}
Consequently, a nontrivial deformation may exist when the crease-crease domain appears in panel $+$ (red domain in Fig.~\ref{fig:sec222}(b)), but the other panel (blue domain in Fig.~\ref{fig:sec222}(b)) will separate at point $S_{c}^{*}$ to avoid the intersection of generators. 

Crease-crease generators also bring extra geometric constraints between linked points. Suppose that a generator starting from $S_1$ ends at $S_2$ (as illustrated in Fig. \ref{fig:sec222}(a)), from $\kappa_m(S_1,v_0(S_1))=\kappa_m(S_2,0)$ we derive that $\kappa(S_2)$ is a function of $(\kappa(S_1),\kappa^{\prime}(S_1),\kappa^{\prime\prime}(S_1),\tau(S_1),\tau^{\prime}(S_1))$.

\begin{figure}
    \centering
    \includegraphics[width=1\textwidth]{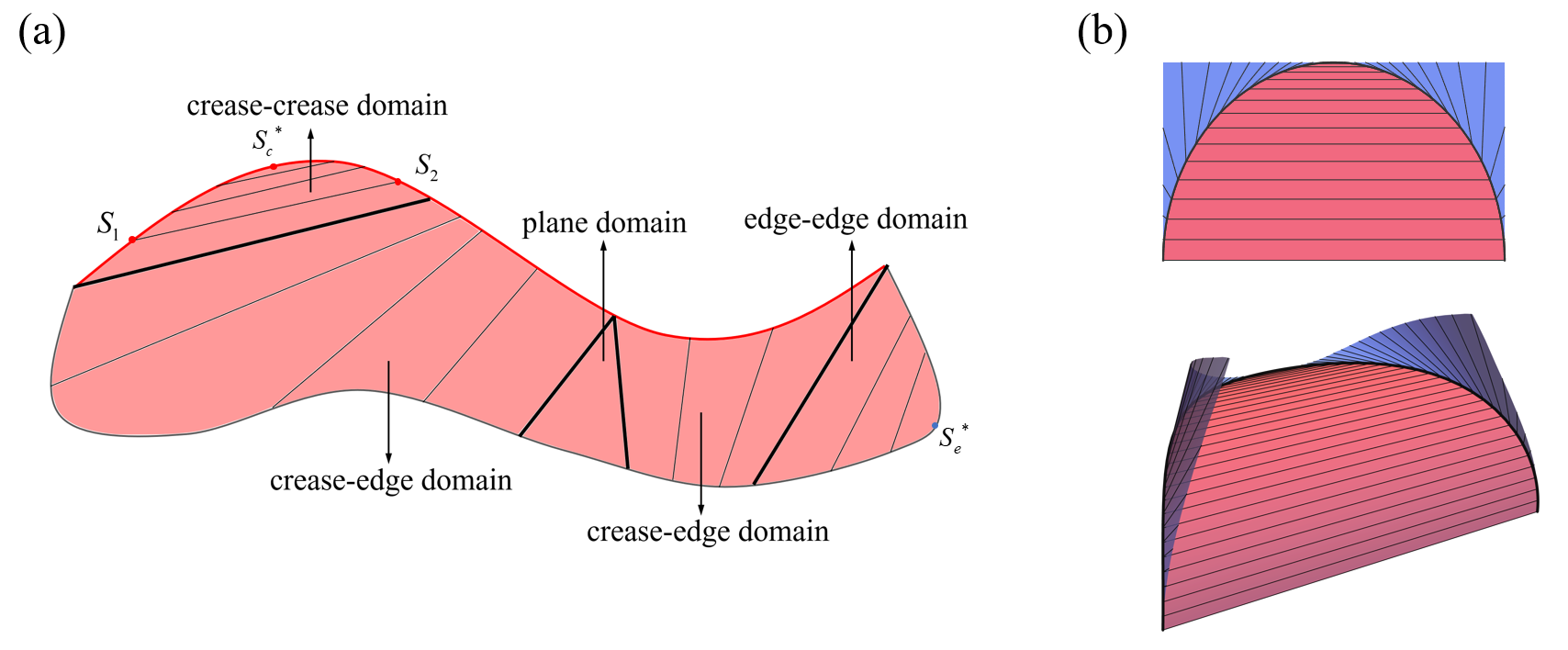}
    \caption{\label{fig:sec222} (a) Generator distribution on a panel with arbitrary reference shape: the red curve represents crease and the thick generators represent boundaries between different domains; (b) curved origami structure with crease-crease domain, and the generator distribution is shown in reference and deformed configuration.}
\end{figure}
\subsection{Mechanics of single-curve-fold origami} \label{sec2.2}
As a conclusion of Section \ref{sec2.1}, the deformation of curved origami depends on the curvature and torsion of the crease. However, in practical scenarios, controlling the curvature and torsion at every point on the crease is impractical. Instead, the folding procedure is conducted by adding mechanical loads or geometric constraints. The mechanical properties of the crease and panels thus play a significant role in the process, and the equilibrium equations of single-curve-fold systems can be derived via energy analysis. 

\subsubsection{Energy distribution in single-curve-fold origami}
The elastic energy of curve-fold origami consists of folding energy concentrated on creases and bending energy distributed in panels.
For folding energy, we assume that folds are modeled as linear elastic torsional springs with direction $\mathbf{t}(S)$. The 
energy of the fold between $(S,S+dS)$ is expressed as
\begin{equation} \label{e17}
  \d E_{fold}=\varepsilon_f(S) \d S = \frac{1}{2}k_c(\phi(S)-\phi_0(S))^2 \d S,
\end{equation}
where $k_c$ is the torsional stiffness per unit length, and $\varepsilon_f(S)$ is the line energy density of the crease. From Eq. \eqref{e14}, $\varepsilon_f(S)$ is a function of $\kappa(S)$
 \begin{equation} \label{e18}
\varepsilon_f(S)=\varepsilon_f(\kappa(S);S).
\end{equation}

For bending energy, the surface energy density at $(S,v)$ of a bending developable surface is proportional to $\kappa_m^2(S,v)$ derived in Eq. \eqref{e13}. Therefore the bending energy between generators starting from $S$ and $S+dS$ is
 \begin{equation} \label{e19}
 \begin{aligned}
 \d E_{bend\pm}&=\varepsilon_{b_{0\pm}}(S)\d S, \\
 &=\frac{D}{2}\int_0^{v_{0\pm}(S)}\kappa_{m\pm}^2(S,v)\d v \d S,\\&=
  \frac{D}{2} \frac{\left(1+t_{\pm}^2(S)\right)^2 \kappa_{n\pm}^2(S)}{\pm t_{\pm}^{\prime}(S)-\left(1+t_{\pm}^2(S)\right) \kappa_{g\pm}(S)} \ln [1\pm t_{\pm}^{\prime}(S)v_{0\pm}(S)-(1+t^2(S)) \kappa_{g\pm}(S)v_{0\pm}(S)] \d S,
 \end{aligned}
\end{equation}
where  $\varepsilon_{b_0}(S)$ is the equivalent line density of bending energy after dimensional reduction and $D$ is the bending modulus of the panel. Note that $v_0(S)$, the projection of the generator length vector onto $\mathbf{e}_2(S)$, is determined by $t(S)$ along with the reference configuration and $\kappa_g(S)$ remains invariant during isometric deformation, therefore from Eq. \eqref{e13},
 $\varepsilon_{b_0}(S)$ is generically expressed as
\begin{equation} \label{e20}
    \varepsilon_{b_0}(S)= \varepsilon_{b_0}(\kappa,\kappa^{\prime},\kappa^{\prime\prime},\tau,\tau^{\prime};S),
\end{equation}
which indicates how bending energy correlates with the geometry of the deformed curve. 

For different domains defined in Section \ref{class}, the bending energy integral is slightly different. Since the bending energy in the edge-edge domain and the crease-crease domain is calculated twice when integrating along the crease or the edge, the equivalent line density energy is augmented to
\begin{equation}
\begin{aligned}
&\varepsilon_{b}\left(S_{c-e}\right)=\varepsilon_{b_0}\left(S_{c-e}\right),\\
&\varepsilon_{b}\left(S_{c-c}\right)=\frac{1}{2} \varepsilon_{b_0}\left(S_{c-c}\right), \\
&\varepsilon_{b}\left(S_{e-e}\right)=\frac{1}{2} \varepsilon_{b_0}\left(S_{e-e}\right),
\end{aligned}
\end{equation}
where $S_{c-c},S_{c-e}$ and $S_{e-e}$ represent the arclength parameter of the crease-crease, crease-edge, and edge-edge domain, respectively. Taking together the integral of $\varepsilon_b(S)$ along the crease and the edge-edge segment yields the bending energy of the whole panel. Therefore, the total elastic energy of the curve-fold origami is
\begin{equation} \label{e22}
  E= E_{fold}+E_{bend}= \int_{crease} (\varepsilon_{b+}+\varepsilon_{b-}+\varepsilon_f) \d S + \int_{e_{+}-e_{+}} \varepsilon_{b+} \d S+ \int_{e_{-}-e_{-}} \varepsilon_{b-} \d S.
\end{equation}
Additionally, we define
\begin{equation} \label{e23}
  \varepsilon(S)= \varepsilon_{b+}(S)+\varepsilon_{b-}(S)+\varepsilon_f(S)=\varepsilon(\kappa,\kappa^{\prime},\kappa^{\prime\prime},\tau,\tau^{\prime};S),
\end{equation}
as the line density of elastic energy for the crease-dependent domain. In conclusion, the energy of single-curve-fold origami systems is expressed as a one-dimensional integral of the function determined by $\kappa, \tau$ of folds and edges. The variation of energy reads the equilibrium equations 
that govern $\kappa$ and $\tau$ of the deformed curve, which will be derived later.
\subsubsection{Equilibrium equations for freely deformed systems}
 We start with the simplest situation where no geometric constraints and mechanical loads are applied to the system. Practically, it describes how an origami structure naturally relaxes to a nonplanar configuration. Note that the edge-edge domain keeps planar when no external loads are added to the panel. Then, the system is only crease-dependent. 
 
 For origami only with the crease-edge domain, the variation of the energy over $\delta \kappa$ and $\delta \tau$ at equilibrium satisfies
\begin{equation}\label{e25}
\begin{aligned}
    \delta E&=\int_{c-e} (\partial_{\kappa}\varepsilon-(\partial_{\kappa^{\prime}} \varepsilon)^{\prime}+(\partial_{\kappa^{\prime\prime}} \varepsilon)^{\prime\prime})\delta \kappa \d S+ \int_{c-e} (\partial_{\tau}\varepsilon+(\partial_{\tau^{\prime}} \varepsilon)^{\prime})\delta \tau \d S\\
    &+(\partial_{\kappa^{\prime}} \varepsilon-(\partial_{\kappa^{\prime\prime}} \varepsilon)^{\prime})\delta \kappa \bigg |_{0}^{S_0} +\partial_{\kappa^{\prime\prime}} \varepsilon \delta \kappa^{\prime}  \bigg |_{0}^{S_0} + \partial_{\tau^{\prime}}\varepsilon\delta \tau  \bigg |_{0}^{S_0}=0,
\end{aligned}
\end{equation}
The variational principal indicates the Euler-Lagrange equations $f=0$ and  $g=0$ along the crease,
where
\begin{equation}\label{e26}
\begin{aligned}
&f(\kappa,\kappa^{\prime},\kappa^{\prime\prime},\kappa^{(3)},\kappa^{(4)},\tau,\tau^{\prime},\tau^{\prime\prime},\tau^{(3)};S)=\partial_{\kappa}\varepsilon-(\partial_{\kappa^{\prime}} \varepsilon)^{\prime}+(\partial_{\kappa^{\prime\prime}} \varepsilon)^{\prime\prime}, \\
&g(\kappa,\kappa^{\prime},\kappa^{\prime\prime},\kappa^{(3)},\tau,\tau^{\prime},\tau^{\prime\prime};S)=\partial_{\tau}\varepsilon+(\partial_{\tau^{\prime}} \varepsilon)^{\prime}, 
\end{aligned}
\end{equation}
and the boundary conditions $\partial_{\kappa^{\prime}} \varepsilon-(\partial_{\kappa^{\prime\prime}} \varepsilon)^{\prime}=0$, $\partial_{\kappa^{\prime\prime}} \varepsilon=0$, $\partial_{\tau^{\prime}}\varepsilon=0$ on both ends $S=0, S=S_0$. 

Considering the crease-crease domain, the corresponding points at the crease on the same generator have additional geometric constraints (Fig.~\ref{fig:sec222}), i.e., only one point is free. Let the segment of the crease-crease domain be given by the arclength interval $(S^*_1, S^*_2)$ with a limiting point $S^*_c$ in between. The integral domain of the total energy can be reduced from $(S^*_1, S^*_2)$ to $(S^*_1, S^*_c)$, given by
\begin{equation}
    \int^{S^*_2}_{S^*_1}\varepsilon(S)\d S=\int^{S^*_c}_{S^*_1}(\varepsilon_{b_{0+}}(S)+\varepsilon_{b_{0-}}(S)+\varepsilon_{f}(S))\d S+\int^{S^*_2}_{S^*_c}(\varepsilon_{b_{0-}}(S)+\varepsilon_{f}(S))\d S=\int^{S^*_c}_{S^*_1}\varepsilon^*(S)\d S,
\end{equation}
where
\begin{equation}\label{evarestar}
\varepsilon^*(S)=\varepsilon^*(\kappa,\kappa^{\prime},...,\kappa^{(4)},\tau,\tau^{\prime},\tau^{\prime\prime},\tau^{\prime\prime\prime};S),
\end{equation}
and the Euler-Lagrange equations $f^*=0, g^*=0$ follow, where
\begin{equation} \label{eccd}
\begin{aligned}
f^*&= \sum_{i=0}^{4}(-1)^i \frac{\d^i}{\d S^i} \partial_{\kappa^{(i)}} \varepsilon^*, \\
g^*&= \sum_{i=0}^{3}(-1)^i \frac{\d^i}{\d S^i} \partial_{\tau^{(i)}} \varepsilon^*.
\end{aligned}
\end{equation}

The Euler-Lagrange equations in different domains are derived in this section, but solving the whole deformation of systems with different domains is challenging and needs further discussion.  The domain distribution is generally unknown before solving the problem, thus all possible distributions and the related solutions should be considered. Furthermore, extra boundary conditions connecting different domains should be satisfied. A further discussion for single-curve-fold origami with different domains is presented in \ref{AppA}, mainly focusing on the additional boundary conditions.
\subsubsection{Equilibrium equations for systems with geometric constraints}\label{sec2333}
In other cases, curved folding is conducted by applying geometric loads at both ends of the crease, such as fixing $\mathbf{r}_c(0)$ and $\mathbf{r}_c(S_0)$ or giving closeness condition $\mathbf{r}_c(0)=\mathbf{r}_c(S_0)$. Since no mechanical loads are applied on the panels, no edge-edge domain participates in the deformation, and the system is crease-dependent. The equilibrium state is also derived using the minimum total potential energy principal. The geometric constraints will introduce Lagrangian multipliers and thus make the Euler-Lagrange equations different from Eq. \eqref{e26}. However, many geometric constraints, including the displacement constraints, are not convenient to be expressed as functions of $\kappa$ and $\tau$. Instead, we variate $E$ by giving a virtual displacement $\delta \mathbf{r}_c$ of the crease. $\delta \kappa$ and $\delta \tau$ are then written as
\begin{equation}\label{e29}
\begin{aligned}
\delta \kappa&=\delta\left(\mathbf{r}_c^{\prime \prime} \cdot \mathbf{n}\right)=\mathbf{n} \cdot \delta \mathbf{r}_c^{\prime \prime},\\
\delta \tau &=\frac{1}{\kappa}\left(\delta\left(\mathbf{r}_c^{\prime \prime \prime} \cdot \mathbf{b}\right)-\tau \delta \kappa\right)=\kappa \mathbf{b} \cdot \delta \mathbf{r}_c^{\prime}-\frac{\kappa^{\prime}}{\kappa^2} \mathbf{b} \cdot \delta \mathbf{r}_c^{\prime \prime}-\frac{\tau}{\kappa} \mathbf{n} \cdot \delta \mathbf{r}_c^{\prime \prime}+\frac{1}{\kappa}\mathbf{b} \cdot \delta \mathbf{r}_c^{\prime \prime \prime}.
\end{aligned}
\end{equation}
$\delta E$ in Eq. \eqref{e25} is then converted into
\begin{equation} \label{e31}
  \delta E= \int_{crease}(g \kappa \mathbf{b} \cdot \delta \mathbf{r}_c^{\prime}-\frac{g\kappa^{\prime}}{\kappa^2} \mathbf{b} \cdot \delta \mathbf{r}_c^{\prime \prime}+(f-\frac{g\tau}{\kappa}) \mathbf{n} \cdot \delta \mathbf{r}_c^{\prime \prime}+\frac{g}{\kappa}\mathbf{b} \cdot \delta \mathbf{r}_c^{\prime \prime \prime})\d S + b.v. ,
\end{equation}
where $b.v.$ stands for boundary values and $\delta \mathbf{r}_c$ satisfies the extensibility condition for the crease $\mathbf{t} \cdot \delta \mathbf{r}_c^{\prime} = 0$. Integrating Eq. \eqref{e31} by parts, we have
\begin{equation}
    \delta E= \int_{crease}[(-f
    \kappa^{\prime}-2\kappa f^{\prime}-g^{\prime}\tau) \mathbf{t}+(-f \tau^2-f\kappa^2+f^{\prime\prime}+g\kappa\tau+\frac{g^{\prime}}{\kappa}\tau^{\prime}+2\tau (\frac{g^{\prime}}{\kappa})^{\prime})\mathbf{n}+(f \tau^{\prime}+2\tau f^{\prime}-(g\kappa)^{\prime}+\frac{g^{\prime}}{\kappa}\tau^2-(\frac{g^{\prime}}{\kappa})^{\prime\prime})\mathbf{b}]\cdot \delta \mathbf{r}_c \d S +b.v.
\end{equation}
However, $\mathbf{t} \cdot \delta \mathbf{r}$, $\mathbf{n} \cdot \delta \mathbf{r}$ and $\mathbf{b} \cdot \delta \mathbf{r}$ are not totally independent. This can be seen by
noting that
\begin{equation}\label{e35}
    (\mathbf{t} \cdot \delta \mathbf{r}_c)^{\prime}=\kappa \mathbf{n} \cdot \delta \mathbf{r}_c+\mathbf{t} \cdot \delta \mathbf{r}_c^{\prime}=\kappa \mathbf{n} \cdot \delta \mathbf{r}_c,
\end{equation}
and therefore
\begin{equation} \label{e36}
    \int_{crease} h\mathbf{n} \cdot \delta \mathbf{r}_c\d S=  \int_{crease} \frac{h}{\kappa}  (\mathbf{t} \cdot \delta \mathbf{r}_c)^{\prime} \d S=-\int_{crease} (\frac{h}{\kappa})^{\prime} \mathbf{t} \cdot \delta \mathbf{r}_c \d S+b.v.,
\end{equation}
for any given $h(S)$.
Combining $\mathbf{t} \cdot \delta \mathbf{r}$ and $\mathbf{n} \cdot \delta \mathbf{r}$ using Eq. (\ref{e36}), Eq. \eqref{e31} is converted to
\begin{equation} \label{e37}
    \delta E= \int_{crease} (f_{gc} \mathbf{t}+g_{gc} \mathbf{b})\cdot \delta \mathbf{r}_c\d S+b.v.=0,
\end{equation}
where
\begin{equation} \label{e38}
\begin{aligned}
     f_{gc}&=f\kappa^{\prime}+2\kappa f^{\prime}+(\frac{-f\tau^2-f\kappa^2+f^{\prime\prime}}{\kappa})^{\prime}+g\tau^{\prime}+2\tau g^{\prime}+(\frac{2\tau (\frac{g^{\prime}}{\kappa})^{\prime}+\frac{g^{\prime}}{\kappa} \tau^{\prime}}{\kappa})^{\prime}=0,\\
     g_{gc}&=\frac{g^{\prime}}{\kappa}\tau^2-(g\kappa)^{\prime}-(\frac{g^{\prime}}{\kappa})^{\prime\prime}+f \tau^{\prime}+2\tau f^{\prime}=0,
\end{aligned}
\end{equation}
with $f,g$ given by Eq. \eqref{e26} for crease-crease domain or \eqref{eccd} for crease-edge domain. The variation method using virtual displacement equivalently changes the variable from arc length parameter $S$ to the vector coordinate $\mathbf{r}_c$, and the geometric constraints are then transformed into boundary conditions. Therefore, to make Eq. \eqref{e37} satisfied for any $\delta \mathbf{r}_c$, the Euler-Lagrange equations under geometric constraints are $f_{gc}=0, g_{gc}=0$.
 
 As a special case, a closed smooth crease automatically satisfies all the boundary conditions, since the variations of the vectors at $\mathbf{r}_c(0)$ and $\mathbf{r}_c(S_0)$  are equal, leading to $b.v.=0$. For systems with small torsion, narrow panels and large crease stiffness, i.e., $\tau\rightarrow 0, v_0\kappa_g\rightarrow 0$ and $D\kappa_g/k_c\rightarrow 0$, Eq. \eqref{e38} degenerates to the equilibrium equations derived in previous work concerning folded annular strips \citep{dias2012geometric}.

In summary, we develop a general framework for understanding the deformation of single-curve-fold origami with arbitrary reference shapes in various mechanical scenarios.

\section{Multi-curve-fold origami theory}\label{sec2new}
Based on the single-curve-fold origami theory derived in Section \ref{sec2}, we move to multi-curve-fold origami. One unique geometric property of multi-curve-fold origami is that the deformation of one fold directly influences the neighboring folds. Therefore, multi-curve-fold origami cannot be theoretically modeled as simple combinations of single-curve-fold origami. In this section, we start with analyzing geometric correlations between folds to derive the geometric mechanics of multi-curve-fold origami.
\subsection{Geometric correlations between curved folds} \label{sec2.3}
 The geometric correlations are built by generators connecting different folds. Suppose that the two generators start at point $S_i$ on crease $i$ and end at point $S_{i+1}$ on crease $i+1$ and point $S_{i-1}$ on crease $i$-1, as shown in Fig. \ref{fig:sec2.31}(a). When the reference shapes and crease patterns are given, $S_{i\pm1}$ are determined by $S_i$ and $\gamma_{i\pm}$, the angle between the crease and the generator. $\gamma(S)$ features the generator distribution by
\begin{equation} \label{e43}
    \gamma_{\pm}(S)=\arccot{t_{\pm}(S)}.
\end{equation}
Therefore, $S_{i\pm1}$, $t_{(i-1)+}(S_{i-1})$ and $t_{(i+1)-}(S_{i+1})$ are expressed in terms of $S_{i}$ and $t_{i\pm}(S)$ when the reference shape and crease pattern are given
\begin{equation} \label{eS}
\begin{aligned}
    &S_{i\pm1}=S_{i\pm1}(t_{i\pm}(S_i);S_i),\\
    &t_{(i\pm1)\mp}(S_{i\pm1}(S_i))=t_{(i\pm1)\mp}(t_{i\pm}(S_i);S_i).
 \end{aligned}
\end{equation}
 \begin{figure}[t!] 
    \centering
    \includegraphics[width=1\textwidth]{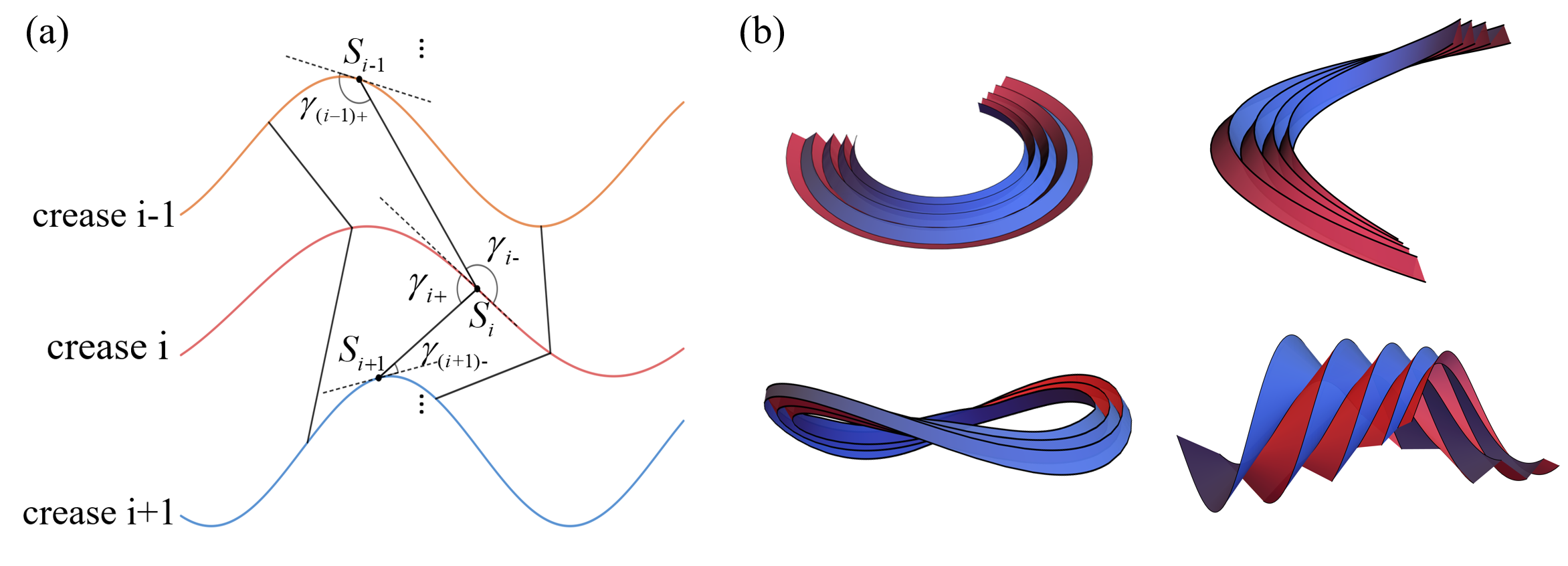}
    \caption{\label{fig:sec2.31} Illustrations of how deformation propagates between creases: (a) crease patterns and geometric relations between neighboring creases illustrated in the reference configuration; (b) examples of numerical experiments. From top left to bottom right: arc strips folded without torsion; arc strips folded with torsion; buckling of closed arc ribbons; a non-Euclidean structure with varying $\kappa_{g\pm}$, folded with torsion.}
\end{figure}
Besides, the correlation of $\kappa_m$ is given by
\begin{equation} \label{e44}
    \kappa_{m{(i\pm1)}\mp}(S_{i\pm1},0)=\kappa_{mi\pm}(S_i,v_{0i\pm}(S_i)),
\end{equation}
which introduces the correspondence of
 $\kappa_n$ between the neighboring creases as
\begin{equation} \label{ek}
    (1+t_{i\pm1}^2(S_{i\pm1}))\kappa_{n(i\pm1)\mp} (S_{i\pm1})
    =\frac{\left(1+t_{i\pm}^2(S_i)\right) \kappa_{ni\pm}(S_i)}{1 \pm v_{0i\pm}(S_i) t_{i\pm}^{\prime}(S_i)-v_{0i\pm}(S_i) \kappa_{gi\pm}(S_i)-v_{0i\pm}(S_i) t_{i\pm}^2(S_i) \kappa_{gi\pm}(S_i)}.
\end{equation}
From Eq. \eqref{e13}, $t$ and $\kappa_n$ are expressed as functions of $\kappa$ and $\tau$, therefore the geometric correlations between neighboring creases can be generally represented in terms of the crease curvature and torsion
\begin{equation} \label{ekt}
\begin{aligned}
\kappa_{i\pm1}(S_{i\pm1})&=\kappa_{i\pm1}(\kappa_i(S_i),\kappa_i^{\prime}(S_i),\kappa_i^{\prime\prime}(S_i),\tau_i(S_i),\tau_i^{\prime}(S_i);S_i),\\
\tau_{i\pm1}(S_{i\pm1})&=\tau_{i\pm1}(\kappa_i(S_i),\kappa_i^{\prime}(S_i),\kappa_i^{\prime\prime}(S_i),\kappa_i^{\prime\prime\prime}(S_i),\tau_i(S_i),\tau_i^{\prime}(S_i),\tau_i^{\prime\prime}(S_i);S_i),
\end{aligned}
\end{equation}
where $S_{i+1}=S_{i+1}(\kappa_i(S_i),\kappa_i^{\prime}(S_i),\tau_i(S_i);S_i)$. Here, the theory is only applicable to the domains that are covered by continuously connected generators between creases.

The propagation of deformation can be visualized in numerical experiments in the following steps. (1) From the given curvature and torsion, the Frenet frame fixed to the crease (defined as crease $i$) is derived, and the generator distribution on panel $i$ is solved from Eq. \eqref{e13}. (2) Darboux frame is then obtained, leading to the parametric equation of panel $i$ and crease $i+1$. (3) Recursively, we can plot the deformed configuration of parts related to the given crease. In some cases where Eq. \eqref{e47} can not be explicitly solved, discrete differential geometry is used in the procedure \citep{bobenko2008discrete, Muller2021}. The results of four numerical experiments are shown in Fig. \ref{fig:sec2.31}(b). This numerical method contributes to the geometric design of curved origami with multiple creases.

\subsection{Mechanics of multi-curve-fold origami}\label{secnew22}
For structures with multiple creases, its elastic energy consists of the folding energy of $N$ creases and the bending energy of $N+1$ panels
\begin{equation} \label{e48}
    E=\sum_{i=0}^{N} E_{b_i}+\sum_{i=1}^{N} E_{f_i},
\end{equation}
The bending energy of panel $i$, $E_{b_i}$, and the folding energy of crease $i$, $E_{f_i}$, are given by
\begin{equation} \label{emultienergy}
\begin{aligned}
    E_{b_i}&=\int_{crease~i} \varepsilon_{b_{i+}}(S) \d S + \int_{crease~i+1} \varepsilon_{b_{(i+1)-}}(S) \d S-\int_{c_i-c_{i+1}} \varepsilon_{b_{i+}}(S) \d S,\\
    E_{f_i}&=\int_{crease~i} \varepsilon_{f_i}(S)\d S,
\end{aligned}
\end{equation}
where $\varepsilon_{f_{i}}$ and $ \varepsilon_{b_{i\pm}}$ are derived in Eq. \eqref{e17} - \eqref{e20}. The system reaches its stable state when the energy minimizes under the constraints expressed in Eq. \eqref{ekt}. With proper boundary conditions, the whole panels are covered by continuously connected generators. The curvature and torsion of crease $N$ are derived from Eq. \eqref{ekt}
\begin{equation}\
\begin{aligned}
        \kappa_{N}(S_N)=\kappa_N(\kappa_1,\kappa_1^{\prime},...,\kappa_1^{(2N-2)},\tau_1,\tau_1^{\prime},...,\tau_1^{(2N-3)};S),\\
    \tau_{N}(S_N)=\kappa_N(\kappa_1,\kappa_1^{\prime},...,\kappa_1^{(2N-1)},\tau_1,\tau_1^{\prime},...,\tau_1^{(2N-2)};S),
\end{aligned}
\end{equation}
and from Eq. \eqref{e18} and \eqref{e20}, Eq. \eqref{e48} is simplified into
\begin{equation}\label{e47}
E=\int_{ crease~1}\varepsilon(\kappa,\kappa^{\prime},...,\kappa^{(2N)},\tau,\tau^{\prime},...,\tau^{(2N-1)};S)\d S.
\end{equation}
According to principals of variation, we obtained the Euler-Lagrange equations $f=0, g=0$ for structures without geometric constraints, where
\begin{equation}\label{emulti1}
\begin{aligned}
f&= \sum_{i=0}^{2 N}(-1)^i \frac{\d^i}{\d S^i} \partial_{\kappa^{(i)}} \varepsilon, \\
g&= \sum_{i=0}^{2N-1}(-1)^i \frac{\d^i}{\d S^i} \partial_{\tau^{(i)}} \varepsilon.
\end{aligned}
\end{equation}
 Following Eq. \eqref{e29} - \eqref{e38}, the Euler-Lagrange equations for systems with geometric constraints can be derived
\begin{equation} \label{emulti2}
\begin{aligned}
    f_{gc}(\kappa,\kappa^{\prime},...,\kappa^{(4N+3)},\tau,\tau^{\prime},...,\tau^{(4N+1)};S)&=f\kappa^{\prime}+2\kappa f^{\prime}+(\frac{-f\tau^2-f\kappa^2+f^{\prime\prime}}{\kappa})^{\prime}+g\tau^{\prime}+2\tau g^{\prime}+(\frac{2\tau (\frac{g^{\prime}}{\kappa})^{\prime}+\frac{g^{\prime}}{\kappa} \tau^{\prime}}{\kappa})^{\prime}=0,\\
    g_{gc}(\kappa,\kappa^{\prime},...,\kappa^{(4N+2)},\tau,\tau^{\prime},...,\tau^{(4N+1)};S)&=\frac{g^{\prime}}{\kappa}\tau^2-(g\kappa)^{\prime}-(\frac{g^{\prime}}{\kappa})^{\prime\prime}+f \tau^{\prime}+2\tau f^{\prime}=0,
\end{aligned}
\end{equation}
where $f$ and $g$ are given by Eq. \eqref{emulti1}.
\subsection{Deformation of curved origami with annular folds}\label{sec2.5}

In this section, we theoretically explain how curved origami structures with annular folds deform within the framework established previously.  The equilibrium configuration minimizes the total elastic energy (folding energy + bending energy) and results in a buckled structure with nonzero torsion. 

Suppose that a generator starts from $S_1$ and ends at $S_2$, as shown in Fig. \ref{fig:sec24}(a). The geometric correlations are given by
\begin{equation}\label{e241}
\begin{aligned}
    &\cos\gamma_{1+}\kappa_{g2}=\cos\gamma_{2-}\kappa_{g1}=\frac{\sin\gamma_{1+}\sin\Delta\theta}{v_0},\\
    &v_0=\sin^2\gamma_{1+}(\sqrt{\kappa_{g1}^{-2}+\frac{1}{\sin^2\gamma_{1+}}(\kappa_{g2}^{-2}-\kappa_{g1}^{-2})}-\kappa_{g1}^{-1}),\\
    &S_2=\frac{\kappa_{g1}}{\kappa_{g2}}S_1+\frac{\Delta\theta}{\kappa_{g2}}.
\end{aligned}
\end{equation}
Thus $t_2,v_0,S_2$ are derived as functions of ($\kappa_{n1},\tau_{n1};S_1$). From Eq. \eqref{ekt}, \eqref{emultienergy} and \eqref{e47}, the energy of the system is derived
\begin{equation}
    E=\int_{crease~1} \varepsilon_1(S_1)\d S_1+\int_{crease~ 2} \varepsilon_2(S_2)\d S_2=\int_{crease~1}(\varepsilon_1(S_1)+\varepsilon_2(S_2(S_1))\frac{\d S_2}{\d S_1})\d S_1.
\end{equation}
Substituting $\varepsilon(S_1)=\varepsilon_1(S_1)+\varepsilon_2(S_2(S_1))\frac{\d S_2}{\d S_1}$ into Eq. \eqref{emulti1} and \eqref{emulti2} yields the equilibrium equations of the system.

Due to the difficulty in numerically solving the full equilibrium equations, we derive an approximate result using the perturbation theory inspired by previous work concerning curved origami with single annular fold \citep{dias2012geometric}, which predicts the amplitude of the normal curvature ($\kappa_n(\pi/(2\kappa_g))-\kappa_n(0)$) and the torsion distribution.
Here we use our generalized energy distribution and keep a higher order term in $\kappa_n$ to reach a better prediction in multi-curve-fold origami. We assume the inner fold has
 $\kappa_{n}(S)=\kappa_{n0}+a_1\cos{(2S\kappa_g)}+a_2\cos{(4S\kappa_g)}$ and  $\tau(S)=b_1\sin{(2S\kappa_g)}+b_2\sin{(4S\kappa_g)}$. From the zeroth order Euler-Lagrange equations, $\kappa_{n0}$ satisfies
\begin{equation}
    \frac{4k^*(2\arctan{\frac{\kappa_{n0}}{\kappa_{g1}}-\phi_0)}}{1+(\frac{\kappa_{n0}}{\kappa_{g1}})^2}-\frac{\kappa_{n0}}{\kappa_{g1}}\ln{\frac{1+2\kappa_{g1}d}{1-\kappa_{g1}d}}=0,
\end{equation}
where $k^*=k_c/D\kappa_{gi}$ is the non-dimensionalized stiffness, chosen to be equal on crease $1$ and $2$. $d$ is the strip width and $\phi_0$ is the rest angle of folds. The closeness of the fold yields a coupling between $\kappa$ and $\tau$, which is satisfied by choosing an appropriate $b_1$. We then vary $a_1, b_2$ to minimize the energy, while $a_2$ is chosen to satisfy the initial value of $\kappa_n$.

We compare the results obtained from the theory and finite element analysis (FEA). The finite element modeling of curved origami is mainly discussed in Section \ref{sec44}, and the qualitative comparison between the FEA result and the real model is shown in Fig. \ref{fig:sec24}(b), indicating that FEA captures the buckling phenomenon well. To illustrate that both geometry and elasticity contribute to the deformation, we calculate two cases with $k^*=k_c/D\kappa_{g}=2,10$ that have the same initial value of $\kappa_n$ ensured by choosing the rest angle $\phi_0$ is changed to ensure that the initial values $\kappa_{n}(0)$ ($\phi_0\approx 0.63\pi, 0.58\pi$). The geometric coefficient is selected as $\kappa_{g1}d=1/10$. The quantatative
 comparison between theoretical and FEA results of $\kappa_{n},\tau$ are shown in Fig. \ref{fig:sec24}(c), where we define $\kappa_n=|\kappa_{n\pm}|,\kappa_g=|\kappa_{g\pm}|$ for convenience. The results show that the approximate theory predicts the distribution of torsion and curvature well, which explains the buckling phenomenon. To obtain the solutions of $\kappa_n$ more precisely, the higher order terms ($\cos(6S),\cos(8S)$) are needed. The results also illustrate that despite the same initial values of $\kappa_n$, the elastic coefficients affect the deformations in the whole field, which is quite different from the single-vertex multi-curve-fold origami that will be discussed shortly. 
\begin{figure}[t!]
    \centering
    \includegraphics[width=1\textwidth]{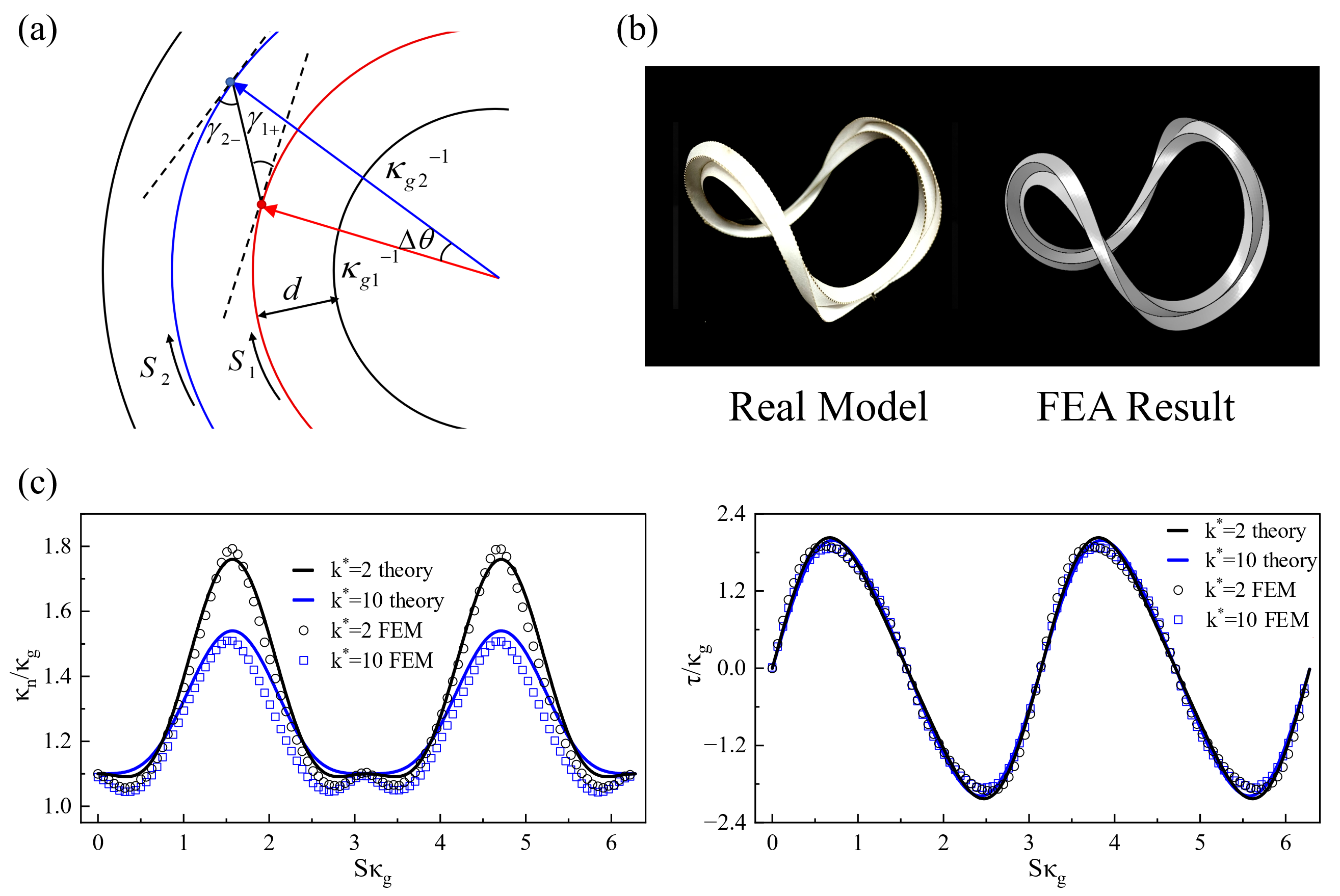}
    \caption{\label{fig:sec24} Deformation of curved origami with annular circular folds: (a) illustrations of geometric variables in the reference shape; (b) deformed configurations achieved in the experiment and finite element analysis; (c) comparison of curvature and torsion of fold between FEA and theoretical results.}
\end{figure}
\section{Single-vertex multi-curve-fold origami theory} \label{sec3}
Among various multi-curve-origami structures, there is a special type with multiple curved folds intersecting at one vertex. Despite a large number of 
single-vertex curved origami, such as those in plants \citep{cheng2023programming} and art design \citep{mitani2019curved}, the theories predicting the deformation still lack. 
In this section, we introduce the extra vertex constraints to the theories in previous sections and develop the theory for single-vertex curved origami, unveiling a universal equilibrium configuration in the region near the vertex. Numerical simulations are performed to validate the theory.

\subsection{Geometric constraints in single-vertex origami} \label{sec3.1}
We start by converting Eq. \eqref{eS} and \eqref{ekt} into explicit forms adaptable to single-vertex curved origami. As shown in Fig. \ref{fig:sec5.22}(a), the correlations between polar coordinates points connected by generators are expressed as
\begin{equation} \label{esanjiao}
\frac{r_1\left(\theta_1\right)}{\sin \left(\gamma_{2-}+\arctan \left(\frac{\d r_2\left(\theta_2\right)}{r_2 \d \theta_2}\right)\right)}=\frac{r_2\left(\theta_2\right)}{\sin \left(\gamma_{1+}+\arctan \left(\frac{\d r_1\left(\theta_1\right)}{r_1 \d \theta_1}\right)\right)}=\frac{v_{0+}}{\sin \left(\gamma_{1+}\right) \sin \left(\alpha+\theta_2-\theta_1\right)},
\end{equation}
where $r(\theta)$ is the crease in polar coordinates. From Eq. \eqref{ek} and \eqref{ekt} revealing geometric correlations between adjacent folds, the crease curvature and torsion of the two points satisfy
\begin{equation} \label{eq049}
\left(1+t_{2-}^2\left(S_2\right)\right) \kappa_{n2-}\left(S_2\right)=\frac{\left(1+t_{1+}^2\left(S_1\right)\right) \kappa_{n1+}\left(S_1\right)}{1+v_{0+}\left(S_1\right) t_{1+}^{\prime}\left(S_1\right)-v_{0+}\left(S_1\right) \kappa_{g1+}\left(S_1\right)-v_{0+}\left(S_1\right) t_{1+}^2\left(S_1\right) \kappa_{g1+}\left(S_1\right)},
\end{equation}
\begin{equation}\label{eq050}
\tau_2(S_2) = -t_{2-}(S_2) \kappa_{n2-}(S_2)-\frac{\kappa_{g2-}(S_2)}{\kappa_{2}^2(S_2)} \kappa_{n2-}^{\prime}(S_2)+\frac{\kappa_{n2-}(S_2)}{\kappa_{2}^2(S_2)} \kappa_{g2-}^{\prime}(S_2),
\end{equation}
where $\d S = \sqrt{r^2+(\frac{\d r}{\d \theta})^2}\d \theta$.
For folds with constant geodesic curvatures, define $\eta=\kappa_{g1}/\kappa_{g2}$, and Eq. \eqref{esanjiao} is explicitly solved by
\begin{equation} \label{e051}
\begin{aligned}
2\theta_{2} & =-\arccos \left(\eta^{-1}(\cos \gamma_{1+}-\cos \left(\gamma_{1+}+2\theta_1\right))+\cos \left(\alpha_1-\gamma_{1+}-2\theta_1\right)\right)-\alpha_1+\gamma_{1+}+2\theta_1, \\
\gamma_{2-} & =\arccos \left(\eta^{-1}(\cos \gamma_{1+}-\cos \left(\gamma_{1+}+2\theta_1\right))+\cos \left(\alpha_1-\gamma_{1+}-2\theta_1\right)\right), \\
v_{0+}& =\frac{2 \eta \sin (\theta_{2}) \sin \left(\alpha_1-\theta_1+\theta_{2}\right) \sin \gamma_{1+}}{\kappa_{g2} \sin \left(\theta_{2}+\gamma_{1+}\right)}.
\end{aligned}
\end{equation}
Considering the N-fold case, the single-vertex curved origami is  constrained by the closeness conditions
\begin{equation}\label{eqclose}
    \begin{aligned}
    \kappa_{N+1}(S)&=\kappa_1(S),\\
    \tau_{N+1}(S)&=\tau_1(S),
    \end{aligned}
\end{equation}
which restrict the creases' configuration spaces. For the axisymmetric deformation, the symmetry condition also needs to be satisfied
\begin{equation} \label{eqperi}
    \begin{aligned}
    \kappa_{i+2}(S)&=\kappa_i(S),\\
    \tau_{i+2}(S)&=\tau_i(S).
    \end{aligned}
\end{equation}
\begin{figure}[t]
    \centering
    \includegraphics[width=1\textwidth]{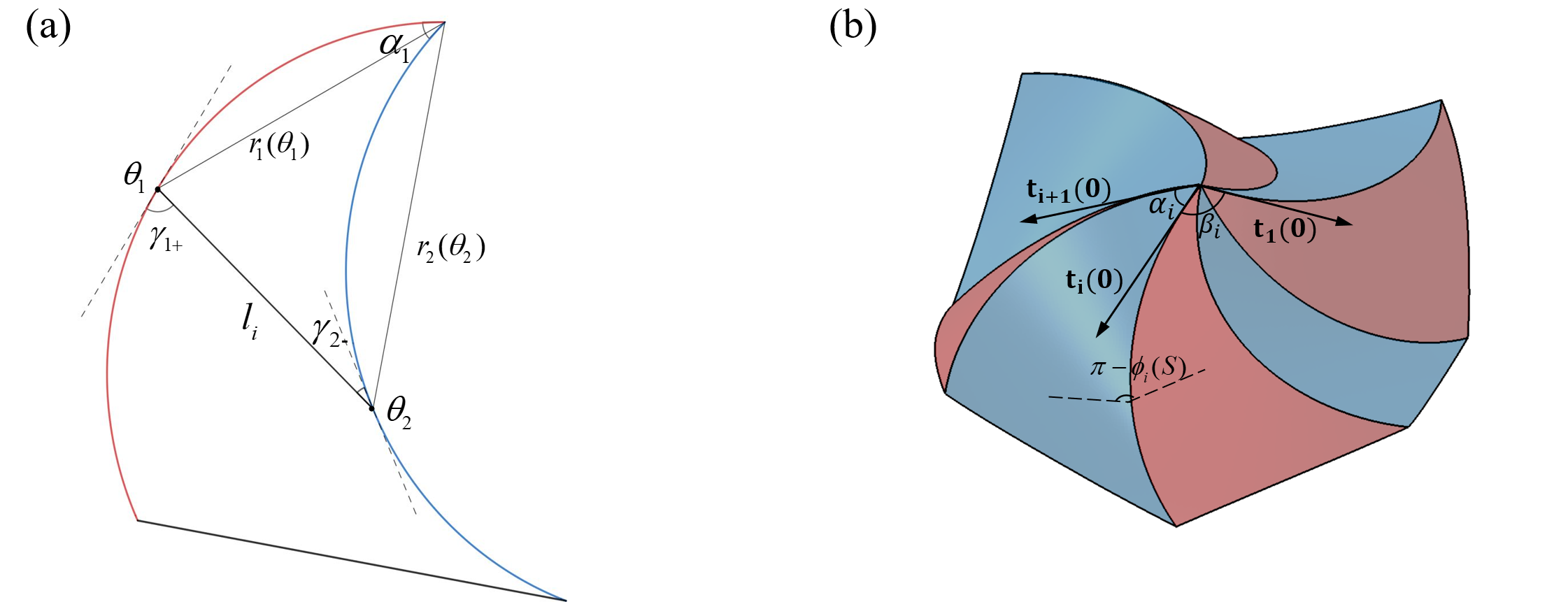}
    \caption{\label{fig:sec5.22} Illustrations for geometric variables: (a) in the panel; (b) at the vertex}
\end{figure}

The tangent vectors of folds intersecting at a vertex lead to extra geometric constraints of the folding angles at $S=0$, as shown in Fig. \ref{fig:sec5.22}(b). We assume that generators cannot intersect at the vertex to form a conical point. Otherwise, $\kappa_m$ reaches infinity as discussed in Section \ref{sec2.1}, and the bending energy diverges. Under the assumption, $\langle \mathbf{t}_i(0),\mathbf{t}_{i+1}(0) \rangle$ stays $\alpha_i$ during isometric deformation. We define $\beta_i=\langle \mathbf{t}_1(0), \mathbf{t}_{i}(0) \rangle$, and according to the previous work on single-vertex origami with straight folds \citep{evans2015lattice}, $\phi_i(0)$, the folding angle of creases at the vertex, are expressed in forms of $\beta_i$ as
\begin{equation} \label{e53}
\begin{aligned}
\phi_1(0)&=\sum_{i=1}^{N-2} \arccos \left(\frac{\cos \alpha_{i+1}-\cos \beta_{i+1} \cos \beta_i}{\sin \beta_{i+1} \sin \beta_i}\right), \\
\phi_2(0)&=\arccos\left(\frac{\cos \beta_2-\cos \alpha_1 \cos \alpha_2}{\sin \alpha_1 \sin \alpha_2}\right), \\
\phi_N(0)&=\arccos\left(\frac{\cos \beta_{N-2}-\cos \alpha_{N-1} \cos \alpha_N}{\sin \alpha_{N-1} \sin \alpha_N}\right),\\
\phi_i(0)&=\arccos\left(\frac{\cos \beta_{i-2}-\cos \alpha_{i-1} \cos \beta_{i-1}}{\sin \beta_{i-1} \sin \alpha_{i-1}}\right)+\arccos \left(\frac{\cos \beta_i-\cos \alpha_i \cos \beta_{i-1}}{\sin \beta_{i-1} \sin \alpha_i}\right).
\end{aligned}
\end{equation}
Therefore $N-3$ $\phi_i(0)$ are independent. Incorporating the rotational symmetry, Eq. \eqref{e53} are simplified as
\begin{equation} \label{eq055}
\begin{aligned}
    [(\sin \frac{\phi_1(0)}{2}\sin\frac{\phi_2(0)}{2})^{-1}-\cot \frac{\phi_1(0)}{2} \cot\frac{\phi_2(0)}{2}]\cos \alpha&=1, \\
\end{aligned}
\end{equation}
which reveals the geometric correlations between the initial folding angles of mountain folds and valley folds. 

\subsection{Near-field deformation} \label{sec3.2}
In Section \ref{sec3.1}, we indicate that the vertex introduces extra periodic constraints (Eq. \eqref{eqclose} and \eqref{eqperi}) and folding angle constraints (Eq. \eqref{e53} and \eqref{eq055}), distinct from multi-curve-fold origami without a vertex. 
Next, we demonstrate that the periodicity at the vertex and the geometric constraints between neighboring creases will bring strong limitations on the configuration space of the origami, leading to distinct behaviors in the \emph{near field} and \emph{far field} defined below.

\textbf{Near field.} The near field is defined as the region near the vertex with limited configuration space. The analytical classification will be given below.

\textbf{Far field.} The far field is the complement of the near field in the origami structure. In this domain, no correlations between creases occur and the configuration space is free.

To study the behaviors in these two domains, we first prove a theorem about the equilibrium configuration in the near field. Then we will derive the exact boundary between these two fields in Section \ref{sec3.3} and study the far-field deformation. 
Suppose that the two generators starting from $S_2$ on crease 2 end at $S_1$ and $S_3$ on crease $1$ and crease $3$, as shown in Fig.~\ref{fig:sec42}(a). The relatively concise theorem on the near-field deformation is based on the following two lemmas. 

\begin{lemma}\label{lemma1}
 If $S_1=S_3$ as illustrated in Fig. \ref{fig:sec42}(a), the curvature and torsion distribution $\kappa_2(S)$ and $\tau_2(S)$ in a finite neighborhood of $S_2$ can be  derived from $\kappa_2(S_2),\tau_2(S_2)$ and $\kappa_1(S_1)$.
\end{lemma}
\begin{proof}
 To guarantee the periodicity constraints, $\kappa_1^{(n)}(S)=\kappa_3^{(n)}(S)$ and $\tau_1^{(n)}(S)=\tau_3^{(n)}(S)$  are satisfied everywhere on crease $1$ and crease 3 with every $n\in \mathbb{Z}$. According to Eq. \eqref{esanjiao} - \eqref{eq050}, the n-th order periodicity constraints can be written as
\begin{equation}\label{etrans}
\begin{aligned}
    h_{\kappa n}(\kappa_2,\kappa_2^{\prime},...,\kappa_2^{(n+2)},\tau_2,\tau_2^{\prime},...,\tau_2^{(n+1)},S_1,S_1^{\prime},...,S_1^{(n)},S_3,S_3^{\prime},...,S_3^{(n)};S_2)&=0,\\
     h_{\tau n}(\kappa_2,\kappa_2^{\prime},...,\kappa_2^{(n+3)},\tau_2,\tau_2^{\prime},...,\tau_2^{(n+2)},S_1,S_1^{\prime},...,S_1^{(n)},S_3,S_3^{\prime},...,S_3^{(n)};S_2)&=0,
\end{aligned}
\end{equation}
where $S_{1,3}^{(n)}=S_{1,3}^{(n)}(\kappa_2,\kappa_2^{\prime},...,\kappa_2^{(n+1)},\tau_2,\tau_2^{\prime},...\tau_2^{(n)};S_2)$. 

Given the initial values $\kappa_2(S_2), \tau_2(S_2)$ and $\kappa_1(S_1)$, the n-th order derivatives $\lbrace\kappa_{2}^{(n)}(S_2)\rbrace$ and $\lbrace\tau_{2}^{(n)}(S_2)\rbrace$ can be solved recursively according to Eq. \eqref{etrans} in the following procedure.
\begin{enumerate}
    \item Since $S_1$ and $S_3$ are determined by $\kappa_2(S_2)$, $\kappa_2^{\prime}(S_2)$ and $\tau_2(S_2)$ (see the equation below Eq. \eqref{etrans}), solving the assumption $S_1=S_3$ yields $\kappa_2^{\prime}(S_2)$ as a function of $\kappa_2(S_2)$ and $\tau_2(S_2)$. 
    \item $\kappa_1(S_1)$ is determined by $\kappa_2(S_2),\kappa_2^{\prime}(S_2),\kappa_2^{\prime\prime}(S_2),\tau_2(S_2)$ and $\tau_2^{\prime}(S_2)$ according to Eq. \eqref{ekt}.
   $h_{\kappa0}=0$ is also an equation for $\kappa_2^{\prime\prime}(S_2)$ and $\tau_2^{\prime}(S_2)$. Solving these two equations yields $\kappa_2^{\prime\prime}(S_2)$ and $\tau_2^{\prime}(S_2)$ as functions of the initial values.
   \item The two equations $h_{\tau0}=0$ and $h_{\kappa1}=0$ give $\kappa_2^{\prime\prime\prime}(S_2)$ and $\tau_2^{\prime\prime}(S_2)$. Then the two equations $h_{\tau1}=0$ and $h_{\kappa2}=0$ give $\kappa_2^{(4)}(S_2)$ and $\tau_2^{(3)}(S_2)$. 
   \item Recursively, $\kappa_2^{(n+3)}(S_2)$ and $\tau_2^{(n+2)}(S_2)$ can be derived from $h_{\tau n}=0$ and $h_{\kappa(n+1)}=0$ by repeating the previous steps.
 \end{enumerate}

The curvature and torsion on the neighboring points are then expressed as
\begin{equation} \label{e056}
\begin{gathered}
\kappa_2\left(S_2+\xi\right)=\sum_{i=0}^n \frac{1}{i !} \kappa_2^{(i)}\left(S_2\right) \xi^i+o\left(\xi^n\right), \\
\tau_2\left(S_2+\xi\right)=\sum_{i=0}^n \frac{1}{i !} \tau_2^{(i)}\left(S_2\right) \xi^i+o\left(\xi^n\right).
\end{gathered}
\end{equation}
Notice that as $n\to \infty$, $o(\xi^n)\to 0$ with finite $\xi$. Thus Lemma $1$ is proved. Therefore, the geometry of point $S_2$ determines the deformation of a finite domain, marked in Fig. \ref{fig:sec42}.
\end{proof}

\begin{figure}[t]
    \centering
    \includegraphics[width=1\textwidth]{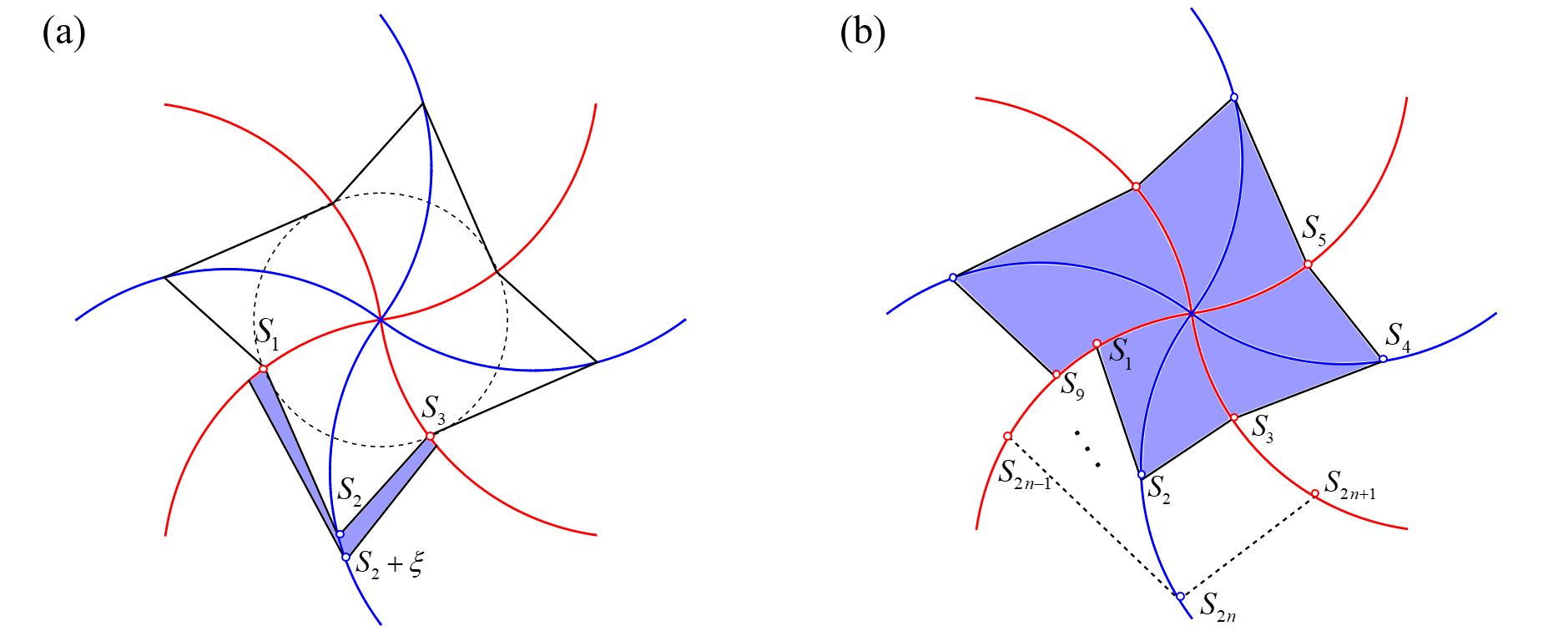}
    \caption{\label{fig:sec42} Two different generator distributions discussed in the proof of (a) Lemma $1$; (b) Lemma $2$. The blue zone represents the geometry-determined domain, i.e., the near field.}
\end{figure}

\begin{lemma}\label{lemma2}
    If $S_1 \neq S_3$ as illustrated in Fig. \ref{fig:sec42}(b), the distribution of $\kappa_2$ and $\tau_2$ in the whole near field can be derived from the distribution of $\kappa_2$ and $\tau_2$ on $(0,S_2)$.
\end{lemma}
\begin{proof}
    We assume that $S_3>S_1$. As shown in Fig. \ref{fig:sec42}(b), if the geometry on $(0,S_2)$ of crease 2 is given, the geometry on $(0,S_3)$ of crease 3 is determined according to Eq. \eqref{eq049} and \eqref{eq050}. The geometry on $(0,S_3)$ of odd folds then determines the geometry on $(0,S_4)$ of even folds. The induction will extend the geometry-determined domain, which continues until $\lim _{n \rightarrow \infty} (S_{n+2}-S_{n})=0$. Then we use Lemma \ref{lemma1} to extend the geometry-determined region, and the procedure continues until a generator reaches the boundary of the structure.
\end{proof}
Based on the above two lemmas, the theorem of the near-field deformation can be proved as follows.
\begin{thm} \label{thm}
     The near-field deformation of a single-vertex curve-fold origami only depends on the curvature at the vertex.
\end{thm}
\begin{proof}
From Eq. \eqref{eq055}, we may derive ${\kappa_1(0)}$ and ${\kappa_3(0)}$ when $\kappa_2(0)$ is given. Since the projection length of generators $v_{0\pm}=0$ at the vertex, the geometric constraints Eq. \eqref{etrans} in Lemma \ref{lemma1} are simplified to
\begin{equation}\label{etrans2}
\begin{aligned}
    &h_{\kappa n}(\kappa_2,\kappa_2^{\prime},...,\kappa_2^{(n+1)},\tau_2,\tau_2^{\prime},...,\tau_2^{(n)},S_1,S_1^{\prime},...,S_1^{(n)},S_3,S_3^{\prime},...,S_3^{(n)};S_2)=0,\\
     &h_{\tau n}(\kappa_2,\kappa_2^{\prime},...,\kappa_2^{(n+2)},\tau_2,\tau_2^{\prime},...,\tau_2^{(n+1)},S_1,S_1^{\prime},...,S_1^{(n)},S_3,S_3^{\prime},...,S_3^{(n)};S_2)=0,
\end{aligned}
\end{equation}
when $S_2=0$. Therefore, as a conclusion of Lemma \ref{lemma1}, the distribution of $\kappa_2$ and $\tau_2$ on a finite range $(0,S_2)$ is determined by $\kappa_2(0)$. As a result of Lemma \ref{lemma2}, the entire near-field deformation can be derived, which is related to $\kappa_2(0)$.
\end{proof}

The theorem indicates that, given the curvature at the vertex, the near-field deformation is independent of other mechanical properties. Moreover, the theorem implies a way of describing the near-field deformation concisely. Specifically, the geometry of the curve-fold origami can be approximately demonstrated by
\begin{equation} 
\begin{gathered}
\kappa\left(S\right)=\sum_{i=0}^n \frac{1}{i !} \kappa^{(i)}\left(0\right) S^i, \\
\tau\left(S\right)=\sum_{i=0}^n \frac{1}{i !} \tau^{(i)}\left(0\right) S^i.
\end{gathered}
\end{equation}
For folds with constant geodesic curvature, the coefficients of the series are easily derived via substituting Eq. \eqref{eq049} - \eqref{e051} into Eq. \eqref{etrans2}.

The near-field theory can be easily generalized to origami with arbitrary curved creases $r_1(\theta)$, not limited to circular arcs. Specifically, Eq. \eqref{esanjiao} is generalized to
\begin{equation}\label{noncircfold}
\sin \left(\gamma_{1+}+\arctan \left(\frac{d r_1\left(\theta_1\right)}{r_1 d \theta_1}\right)-\left(\alpha+\theta_2-\theta_1\right)\right) r_2\left(\theta_2\right)=\sin \left(\gamma_{1+}+\arctan \left(\frac{d r_1\left(\theta_1\right)}{r_1 \d \theta_1}\right)\right) r_1\left(\theta_1\right),
\end{equation}
where $\theta_1$ and $\theta_2$ are the corresponding parameters at the generator connecting two neighboring creases. We could then
compute $\theta_2^{(n)}(\theta_1)$ at the vertex from Eq. \eqref{noncircfold}. The coefficients of the series are thus derived from Eq. \eqref{etrans2} to obtain the deformation of the full domain. Furthermore,
the symmetric deformation condition Eq. \eqref{eqperi} can also be relaxed. Notice that the closeness condition Eq. \eqref{eqclose} will bring intrinsic periodicity. We may follow similar steps to prove the near-field theory. However, in the asymmetric case, $\kappa_n(0)$ is not sufficient. Instead, the higher order terms $\kappa_n'(0), \kappa_n''(0) \dots$ are needed.

\subsection{Far-field deformation} \label{sec3.3}
In Lemma \ref{lemma2}, we show that the deformation is determined by the near-field theory until a generator ends at an edge. In other words, the generator distribution determines the boundary between the near-field domain and the far-field domain. Figure \ref{fig:sec43} shows two possible generator distributions that are distinguished by the relationships between $S_1(S_0), S_3(S_0)$ and $S_0$. 

\begin{figure}[b!]
    \centering
\includegraphics[width=0.95\textwidth]{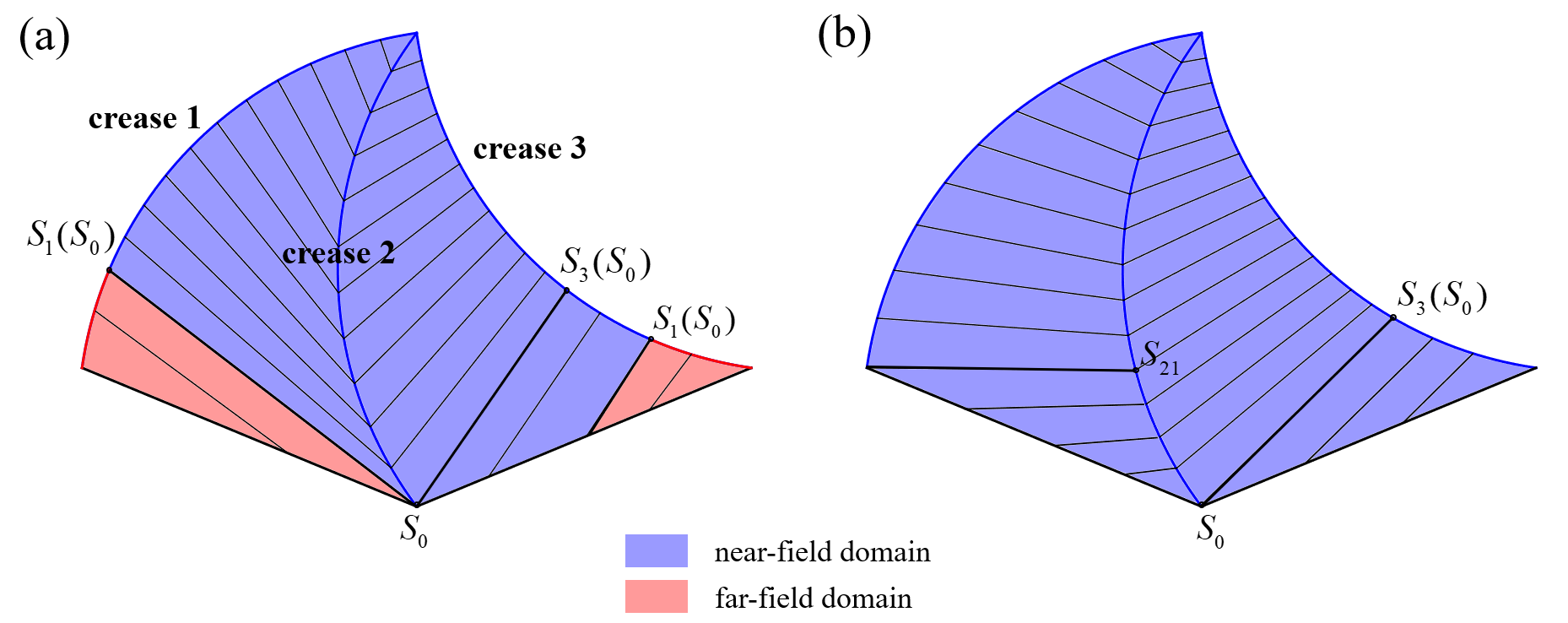}
    \caption{\label{fig:sec43} Two possible near-field domain and far-field domain distributions based on the generator distribution: (a) crease $2$ and blue regions are determined by near-field theory; (b) all folds and whole panels are determined by near-field theory.}
\end{figure}
\textbf{Situation A.} As shown in Fig. \ref{fig:sec43}(a), we have $S_1(S_0)<S_0$ and $S_3(S_0)<S_0$. Therefore, all generators starting from crease $2$ (defined in Fig. \ref{fig:sec43}(a)) on both panels end at the neighboring creases, leading to the conclusion that segment $(0,S_0)$ on crease 2 and segments $(0,S_1(S_0))$ on crease $1$ and 3 are determined by the near-field theory.  The near-field and far-filed domains can be identified. 

The deformation for Situation A can be derived as follows. Given $\kappa_n(0)$, the boundary $S_1(S_0)$ and the corresponding values $\kappa_n(S_1),\kappa_n^{\prime}(S_1),\tau(S_1)$  are determined by the near-field theory. The deformation of the far-field region is then solved via energy minimization.  For the systems without geometrical loads, the far-field regions deform freely, following Eq. \eqref{e26} and the initial values $\kappa_n(S_1),\kappa_n^{\prime}(S_1),\tau(S_1)$. For systems with geometrical loads, $\mathbf{r_c}(S_1)$ is determined by the near-field theory, and the far-field deformation is thus given by Eq. \eqref{e38}.

\textbf{Situation B.} As shown in Fig. \ref{fig:sec43}(b), we have $S_1(S_0)>S_0$ and $S_3(S_0)<S_0$. Assuming that $S_1(S_{21})=S_0$, the segment $(0,S_{21})$ on crease $2$ satisfies the near-field theory. The segments $(0,S_0)$ on crease $1$ and $3$ are then determined, leading to the conclusion that the whole structure is determined by the near-field theory. 

With both near-field and far-field theory, the deformation of single-vertex curved origami is derived.
As for the situation with $S_1(S_0)>S_0$ and $S_3(S_0)>S_0$, it is included in situation A by selecting crease $1$ or crease $3$ as the new crease $2$. Situations with $S_1(S_0)<S_0$ and $S_3(S_0)>S_0$ are included in situation B due to symmetry. Therefore, the two cases include all possibilities. As a conclusion, If generators distribute continuously, at least one fold (crease $2$) is completely in the near-field domain.

\subsection{Finite element modeling and numerical results of single-vertex curved origami}
Numerical simulations of single-vertex curved origami are employed in this section to validate the single-vertex curved origami theory.
\subsubsection{Finite element modeling of curve-fold origami}\label{sec44}
In finite element analysis, the origami panels are modeled as shells, and the elastic folds are modeled as discrete torsional springs \citep{FEM,water}. In the following analysis, we generate finite element models using the commercial software Abaqus. The origami panels are modeled as S4R shell elements. The curved folds are modeled as join-rotation connectors with stiffness $k_{1}=k_cS_0/N$ along the tangent of the folds, where $S_0$ is the fold length and $N$ is the number of nodes on each fold. The rest angle of curved folds is implemented by adding a connector moment of $m_{1}=\pm k_1\phi_0$ to the connectors on mountain and valley folds.
We then simulate the deformation of the single-vertex curved origami with different crease patterns. The qualitative comparison between FEA results and real models is shown in Fig. \ref{fig:sec51}, illustrating good agreement. $R$ represents the distance between the starting and ending points of the creases.

\begin{figure}[b!]
    \centering
    \includegraphics[width=0.6\textwidth]{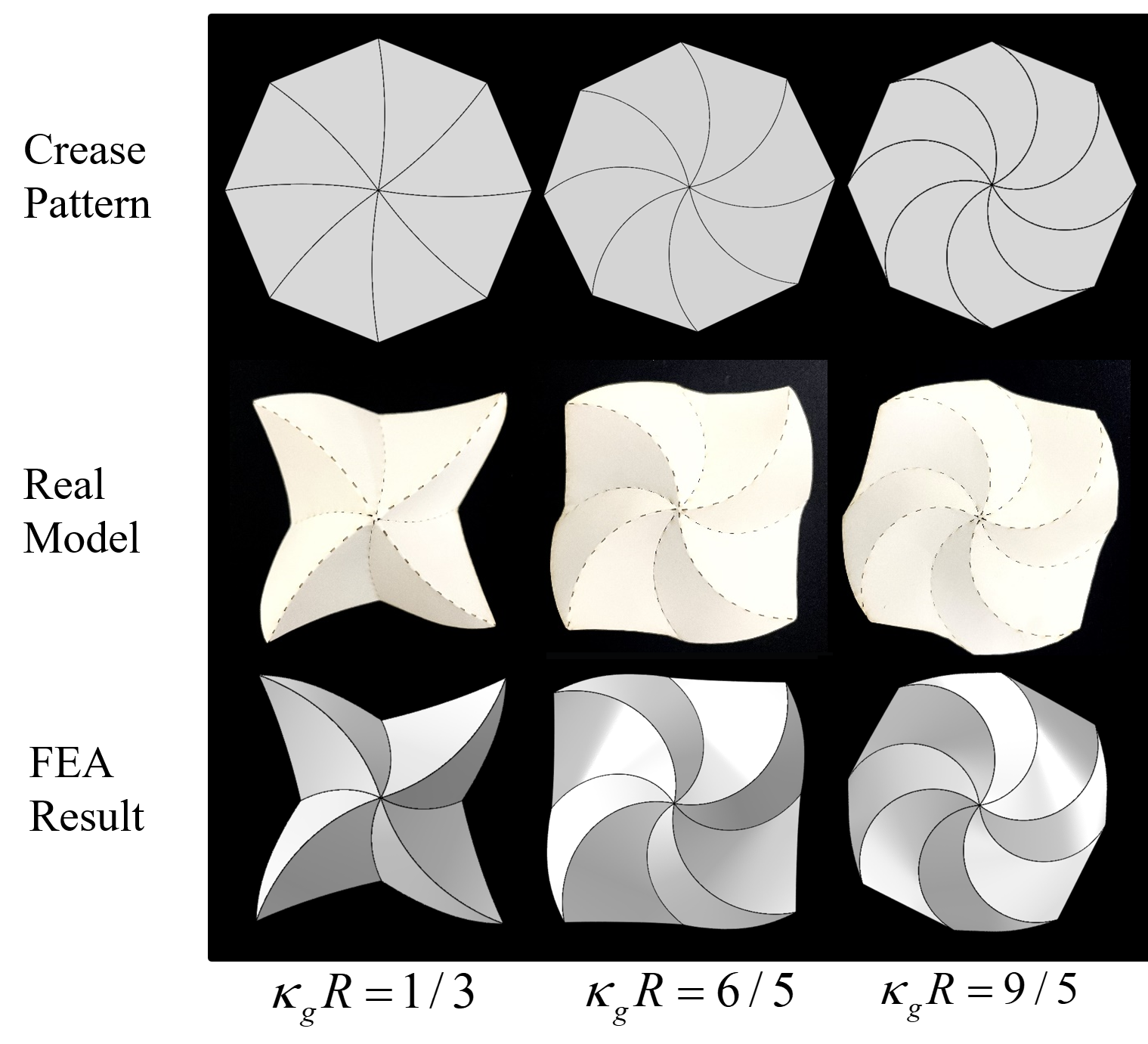}
    \caption{\label{fig:sec51}Deformed and reference configurations of the single-vertex curved origami with folds of different curvatures. }
\end{figure}
\subsubsection{Comparison between theoretical solutions and numerical results}
\begin{figure}[b!]
    \centering
    \includegraphics[width=1\textwidth]{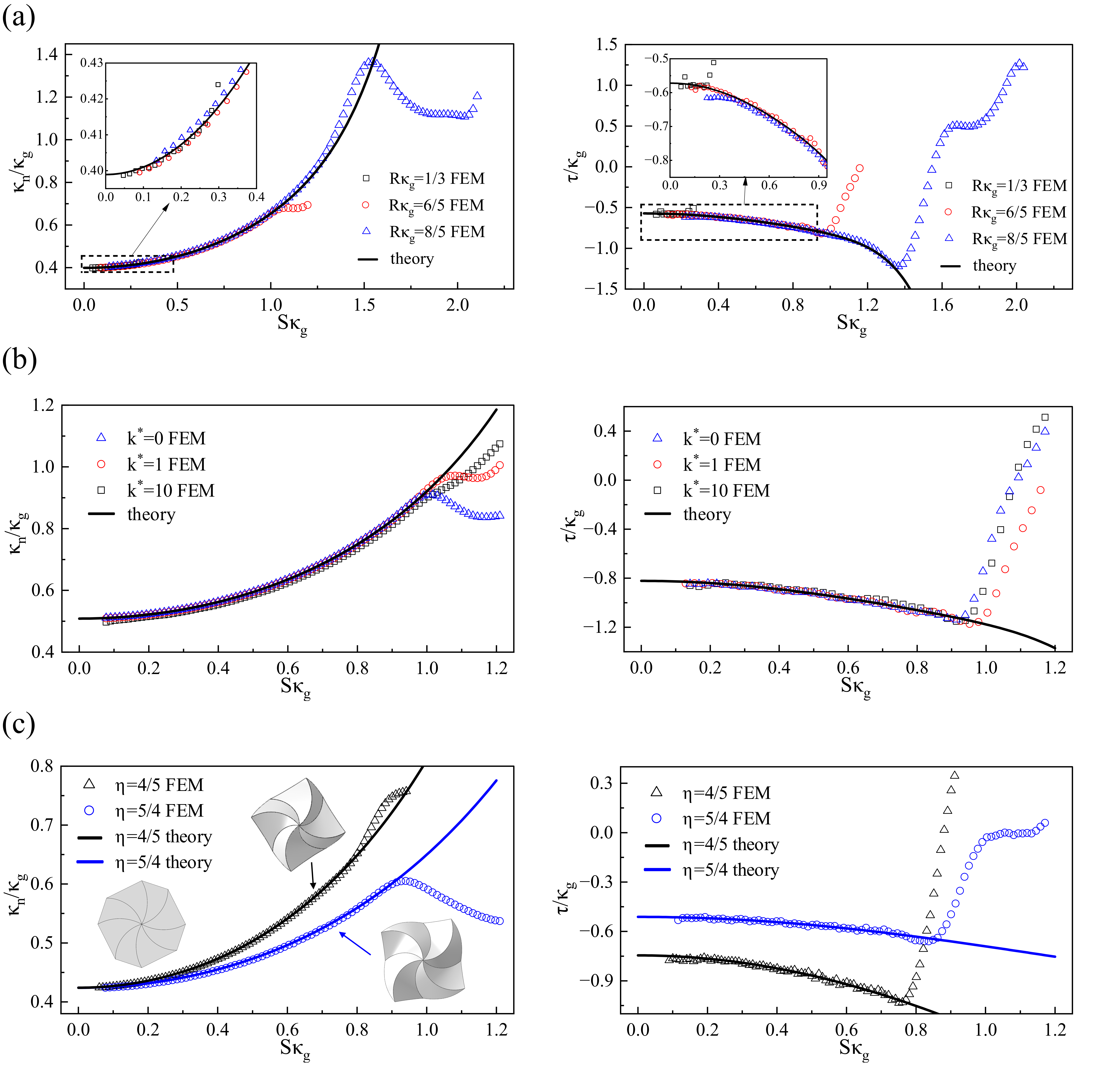}
    \caption{\label{fig:sec52}Comparisons between the near-field theoretical solutions and the numerical results: (a) folds with different $\kappa_g$; (b) folds with different $k^*=k_c/D\kappa_g$; (c) folds with distinct neighboring $\kappa_g$.}
\end{figure}
Here we compare the normal curvature and the torsion distribution of folds derived from theory and FEA. In FEA,
the coordinates of the deformed folds are extracted to compute the discrete curvature and torsion. In the following situations, if crease $2$ is selected as the valley fold, we have $S_3(S_0)<S_0$ derived from the near-field theory. As predicted in Section \ref{sec3.3}, the deformation of the valley folds is completely determined by the near-field theory, which will be examined in the following comparisons. To validate the universality of the near-field theory in Section \ref{sec3.2}, we compare the near-field solutions with numerical results of the folds, where we define $\kappa_n=|\kappa_{n\pm}|,\kappa_g=|\kappa_{g\pm}|$ and the variables $S,\kappa_g,\kappa_n,\tau$ are non-dimensionalized to $S\kappa_g,R\kappa_g,\kappa_n/\kappa_g,\tau/\kappa_g$. The following situations are discussed.

\textbf{Folds with different $\boldsymbol{\kappa_g}$}. We model the arc-fold structures with different geodesic curvatures $R\kappa_g=1/3,6/5,9/5$, and take the rest angles $\phi_0\approx 0.45\pi,0.77\pi,\pi$ such that the deformed configurations have the same vertex folding angles $\phi_1(0)$. Figure \ref{fig:sec52}(a) shows that the theoretical and FEA results reach good agreement in most of the near-field domain.

\textbf{Folds with different $\boldsymbol{k_c}$}. We model arc-fold structures with different non-dimensionalized stiffness $k^*=k_c/D\kappa_g=10,1,0$. For $k^*=10,1$, the folding process is simulated by giving the rest angles $\phi_0\approx 0.55\pi,\pi$, while folding structure with zero $k^*$ is achieved via applying a vertical load $\Delta z\approx 0.63R$ at the vertex. The results in Fig. \ref{fig:sec52}
validate that the deformations are independent of elastic coefficients in a large range once the initial folding angle $\phi_1$ is given.

\textbf{Folds with distinct neighboring $\boldsymbol{\kappa_g}$}. We model the structures with distinct neighboring geodesic curvature $\kappa_{g1},\kappa_{g2}$, where $\kappa_{g2}/\kappa_{g1}=4/5$. Let $\eta$ be the ratio between the geodesic curvatures of the mountain and valley folds. Here, we simulate the cases that have $\eta=\kappa_{g1}/\kappa_{g2} = 4/5$ and  $\eta=\kappa_{g2}/\kappa_{g1}=5/4$, by applying vertical displacements $\Delta z \approx 0.69R, 0.75R$ at vertices. The vertical displacements are chosen to make the initial folding angles identical for the two cases. The results in Fig. \ref{fig:sec52}(c) illustrate that the near-field theory works well in a large range.

The above results validate that the near-field solutions Eq. \eqref{e056} derived from Eq. \eqref{etrans2} reach satisfactory agreement with simulations for various situations. It indicates that the near-field deformation is solely dependent on the initial value $\kappa_n(0)$ and is independent of the geodesic curvatures, elastic coefficients and deforming processes. In contrast, the deformation of multi-curve-fold origami without vertices is governed by the Euler-Lagrange equations, which relies on the elastic coefficients, as shown in Section \ref{sec2.5}. We prove in Section \ref{sec3.2} that, the significant difference owes to the periodicity at the vertex. However, a little unsatisfactory part remains. As analyzed in Section \ref{sec3.3}, the deformation of the valley fold should be completely determined by the near-field theory if the generators distribute continuously, but Fig. \ref{fig:sec52} shows disagreement at the end of the folds. We discuss this observation further in the next section.

\subsubsection{Further discussion on the far-field solutions}
We anticipate that the discontinuities in $\tau$ and $\kappa'_n$ cause the
deviation between FEA and theoretical results in Fig.~\ref{fig:sec52}. As discussed in Section. \ref{sec2.1}, discontinuities occur when $\tau$ and $\kappa_n^{\prime}$ jump, which will form a planar domain. The new generator distribution is shown in Fig. \ref{fig:secfar}(b), which is chosen to satisfy the energy density contour in Fig. \ref{fig:secfar}(a). Such distribution only occurs when the periodic symmetry slightly breaks, since it contradicts with Lemma \ref{lemma2} proved based on a strictly satisfied periodicity.  The good agreement in the near-field region indicates that we may assume the broken symmetry does not affect the low-order (depends on $S_0\kappa_g$) coefficients of the near-field solutions. Therefore, on $(0,S_{21})$ and $(0,S_{3})$, $\kappa$ and $\tau$ are determined by the near-field theory, but the geometry on $(S_3(S_{22}),S_1(S_{21}))$ cannot be determined by the geometry on $(0,S_{21})$ due to the symmetry breaking. Since $t_{2+}$ is continuous at $S_{21}$ and $t_{2-}$ is continuous at $S_{22}$, from Eq. \eqref{e13} the geometric coefficients near $S_{21}$ and $S_{22}$ satisfy 
\begin{equation} \label{jump}
\begin{aligned}
    \tau(S_{21}-)-\frac{\kappa_{n+}^{\prime}(S_{21}-)}{\kappa^2(S_{21})}\kappa_{g+}&=\tau(S_{21}+)-\frac{\kappa_{n+}^{\prime}(S_{21}+)}{\kappa^2(S_{21})}\kappa_{g+},\\
       \tau(S_{22}-)+\frac{\kappa_{n-}^{\prime}(S_{22}-)}{\kappa^2(S_{22})}\kappa_{g-}&=\tau(S_{22}+)+\frac{\kappa_{n-}^{\prime}(S_{22}+)}{\kappa^2(S_{22})}\kappa_{g-}.
\end{aligned}
\end{equation}
\begin{figure}[b!]
    \centering
    \includegraphics[width=1\textwidth]{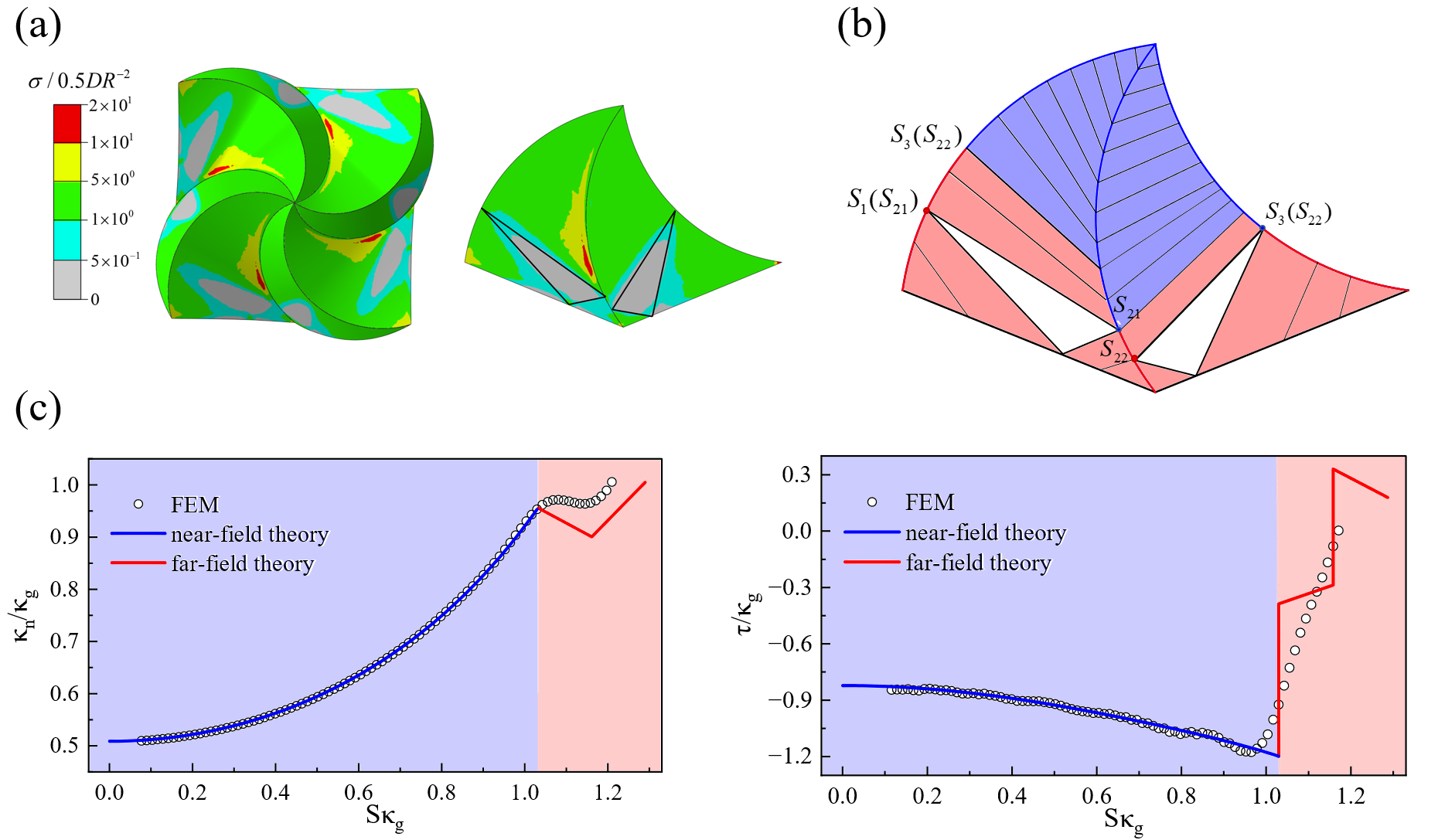}
    \caption{\label{fig:secfar}Further discussion on far-field theories. (a) Energy density contour, where grey triangles represent plane domain; (b) generator distribution after the symmetry breaking; (c) comparison between theoretical solutions and FEA results.}
\end{figure}
From the energy density contour and the results in the previous subsection, we assume that for the case $S_0\kappa_g=1.2$, discontinuities on crease $2$ occur at $S_{21}=0.8S_0$ and $S_{22}=0.9S_0$. $S_3(S_{22})$ and $S_1(S_{21})$ are then derived from the near-field theory, since numerically we observe that changes in high-order terms affect little to both $t_{2-}(S)$ on $(0,S_{21})$ and $t_{3-}(S)$ on $(0,S_3(S_{22}))$. The jumps in $\tau$ and $\kappa_n'$ should minimize the bending and folding energy of the far field, where the bending energy of crease-edge domains is given by Eq. \eqref{e19}, with the projection length of generators $v_{0\pm}$ computed as
\begin{equation}\label{v0}
\begin{aligned}
        v_{0\pm}\kappa_g=\pm 2 \frac{\sin(\frac{S^*_0-S^*}{2})\cos(\frac{S^*\mp\alpha}{2})\sin(\gamma_\pm)}{\cos(\gamma_\pm+S^*+\frac{S^*_0\mp\alpha}{2})}.
\end{aligned}
\end{equation}

To give a first approximate solution for the far-field deformation, we assume that in the far field, $\kappa_n$ of mountain and valley folds are linear in arclength. For $k^* \approx 1$, the bending and folding energies are comparable, therefore we assume that the changes in bending and folding energies caused by 
the jump in $\kappa^{\prime}_n$ cancel each other, and the change in generator distributions, i.e., $t$, dominates the total energy change. 
Specifically, the jumps of $t$ should minimize the zero-order energy given by the multiplication of area and energy density at a single point
\begin{equation}
    E_0=\frac{D}{2}[(1+t_{2+}^2(S_{22}+)\kappa_{n2})^2A_{2+}+(1+t_{2-}^2(S_{21}+)\kappa_{n2})^2A_{2-}+(1+t_{1+}^2(S_1(S_{21})+)\kappa_{n1})^2A_{1+}+(1+t_{3-}^2(S_3(S_{22})+)\kappa_{n3})^2A_{3-}]
\end{equation}
where $A_{i\pm}$ represents the surface area of the crease-edge domain, which is determined by $v_{0\pm}$ in Eq. \eqref{v0} once $t$ is given. In zero-order energy, we assume
\begin{equation}
    \frac{\kappa_{n2}(S_{21})}{\kappa_{n1}(S_1(S_{21}))}=\frac{\kappa_{n2}(S_{22})}{\kappa_{n3}(S_3(S_{22}))}=\frac{\kappa_{n2}(S_{21})}{\kappa_{n1}(S_3(S_{22}))}.
\end{equation}
Thus the ratio can be derived from the near-field theory. Since generators starting from $S_1(S_{21})$ and $S_{21}$ form a triangle, $t_{2-}(S_{21}+)$ couples with $t_{1+}(S_1(S_{21})+)$. Similarly $t_{2+}(S_{22}+)$ couples with $t_{3-}(S_3(S_{22})+)$. We could then derive $\Delta t_{2\pm}$ via separately minimizing the energy $E_0$ on the two panels. From Eq. \eqref{e13} and \eqref{jump}, the jump of $\kappa_n^{\prime}$ satisfies
\begin{equation}
\Delta\kappa_n^{\prime}=\pm\frac{\kappa_n\kappa^2\Delta t_{2\pm}}{2\kappa_g}.
\end{equation}
Therefore, with $\Delta t_{2\pm}$ derived by energy minimization, $\Delta \kappa_n'$ at $S_{21}$ and $S_{22}$ is derived. According to the linear assumption, the far-field curvature is then given approximately. The jumps of $\tau$ are given as well. $\tau(S)$ in $(S_{21},S_{22})$ is given by Eq. \eqref{e13}
\begin{equation}
    \tau=-t_{2+}\kappa_{n+}+\frac{\kappa_{n+}^{\prime}}{\kappa^{2}}\kappa_{g+},
\end{equation}
where $\kappa_{n+}$ is given by the revised far-field theory and $t_{2+}$ is given by the near-field theory. On $(S_{22},S_0)$, $\tau$ is chosen to approximately minimize the energy of the segment.
Fig. \ref{fig:secfar}(c) shows the comparison between the theoretical and numerical results, validating the revised far-field theory. The geometrical and elastic properties of the structure are  $R\kappa_g=6/5, k^*=1, \phi_0=\pi$.

\section{Conclusion} \label{sec:conclusion}
In this work, a comprehensive framework for multi-curve-fold origami is established. This framework theoretically explains how geometry and elasticity determine the deformation of curve-fold origami structures. As an example of multi-curve-fold origami with periodicity, the theory of single-vertex curved origami is established, unveiling a striking vertex-constrained universal equilibrium configuration. Our main contributions are summarized as follows:

(1) We characterize all possible generator distributions and the corresponding energy distributions in single-curve-fold origami with wide panels. We then derive the equilibrium equations that predict how single-curve-fold structures with arbitrary reference tiles deform via energy minimization. Our theory removes the limitations on the panel width and generator distribution.

(2) We derive the geometrical correlations between neighboring creases in multi-curve-fold origami. Based on this, we extend the single-curve-fold origami theory to the multi-curve-fold origami theory. The theory is used to solve the deformation of curved origami with annular creases. Numerical simulations validate the theory.

(3) We derive the single-vertex curved origami theory. We prove that the periodicity at the vertex strongly constrains the configuration space, yielding a universal equilibrium shape at the near-field domain, regardless of the mechanical properties. In contrast, 
the far-field theory, derived by energy minimization, is dependent on the mechanical properties. Numerical simulations are conducted, showing good agreement with theoretical predictions.

 We believe that our generalized multi-curve-fold origami theory, including the vertex-constrained universality, can extend the understanding of the physics of the curved origami system and provide new insights into the shape programming. We anticipate that our work can contribute to the design of complex curved origami structures in the fields of robotics, metamaterials and architectures.
\section*{CRediT authorship contribution statement}
\textbf{Zhixuan Wen:} Conceptualization, Methodology, Software, Validation, Formal analysis, Writing – original draft. \textbf{Pengyu Lv:} Methodology, Writing– review \& editing. \textbf{Fan Feng:} Conceptualization, Investigation, Methodology, Software, Writing– review \& editing. \textbf{Huiling Duan:} Conceptualization, Methodology, Writing – review \& editing, Supervision, Project administration, Funding acquisition. 
\section*{Declaration of competing interest}
The authors declare that they have no known competing financial interests or personal relationships that could have appeared to influence the work reported in this paper.
\section*{Data availability}
No data was used for the research described in the article.
\section*{Acknowledgement}

\appendix
\section{Further discussion on mechanics of single-curve-fold systems with different domains}\label{AppA}

\setcounter{table}{0}   
\setcounter{figure}{0}
\setcounter{equation}{0}

\renewcommand\thetable{A.\arabic{table}}
\renewcommand\thefigure{A.\arabic{figure}}
\renewcommand\theequation{A.\arabic{equation}}

In this appendix, we analyze the equilibrium equations and the boundary conditions for systems with different domains. While the domain classification is usually unknown before solving the equations, we should analyze all possible domain distributions and compare the related elastic energies. The distribution with the minimum energy is the real solution. 

As stated in Section \ref{sec2.2}, the folding processes discussed in this paper will not introduce edge-edge domains. As analyzed in Section \ref{class}, the crease-crease domain only appears in systems with separate panels. We then give illustrations of possible generator distributions in a system that may have the crease-crease domain, the crease-edge domain and the plane domain, as shown in Fig. \ref{fig:A1}. Let the starting point of the possible crease-crease domain $S_1^{*}$ be $S=0$. The value of $t(0)=(\kappa_{g}(0)\kappa_{n}^{\prime}(0)/\kappa^2(0)-\kappa_{n}(0)\kappa_{g}^{\prime}(0)/\kappa^2(0)+\tau(0))/\kappa_n(0)$ determines the type of the generator distribution. In the following sections, we give the total energy variations for both situations and then discuss the equilibrium equations and boundary conditions.

\begin{figure}[t!]
    \centering
    \includegraphics[width=1\textwidth]{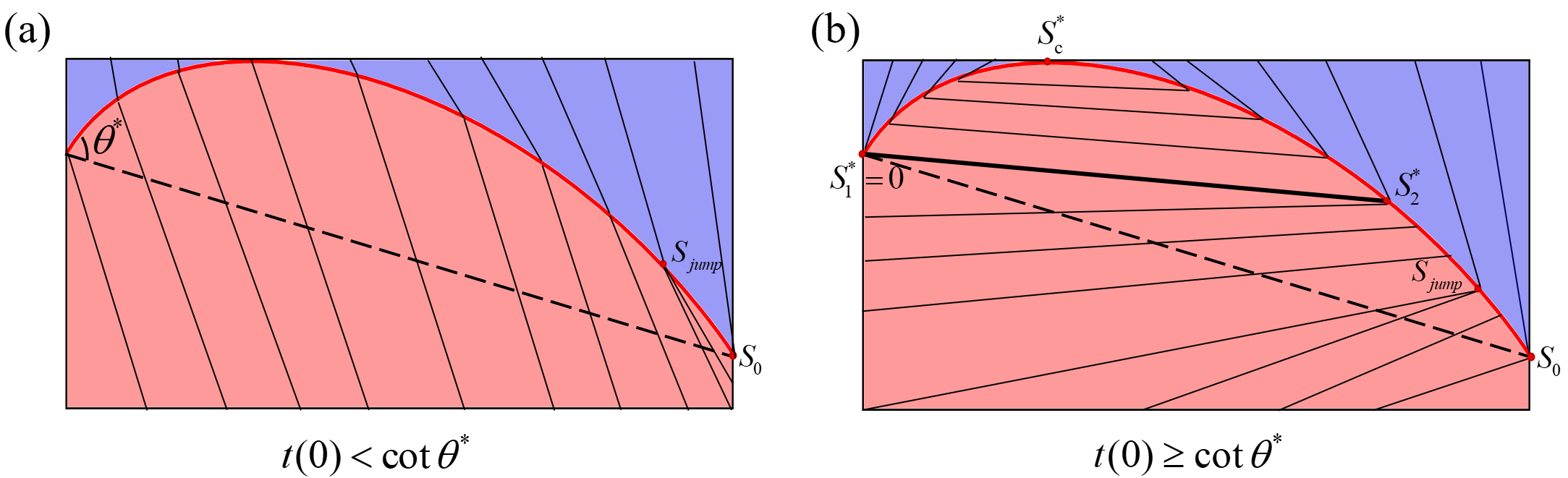}
    \caption{\label{fig:A1}Two possible generator distributions in a single-curve-fold system (panel $+$: red, panel $-$: blue). (a) If $t(0)<\cot{\theta^*}$, the system is covered with crease-edge domain and plane domain. (b) If $t(0)\geq \cot{\theta^*}$, the system is covered with crease-crease domain, crease-edge domain and plane domain. The black thick generator represents the boundary of the crease-crease domain.}
\end{figure}

\subsection{Energy variations}
\boldmath{$t(0)<\cot{\theta^*}$}.
In this situation, the system is covered with crease-edge domains and possible plane domains. Since the plane domain has no elastic energy, the elastic energy of the single-curve-fold system is\unboldmath
\begin{equation}\label{a1}
   E=\int_0^{S_0} \varepsilon(S)\d S,
\end{equation}
where $\varepsilon$ is defined in Eq. \eqref{e23}. Considering possible jumps in geometric variables caused by non-smooth edges and noncontinuous geodesic curvatures at $S_{jump}$, the energy variation in Eq. \eqref{e25} is transformed into
\begin{equation}\label{eA2}
\begin{aligned}
    \delta E&=\int_{0}^{S_0} f\delta \kappa \d S+ \int_{0}^{S_0} g \delta \tau \d S +(\partial_{\kappa^{\prime}} \varepsilon-(\partial_{\kappa^{\prime\prime}} \varepsilon)^{\prime})\delta \kappa \bigg |_{0}^{S_0} +\partial_{\kappa^{\prime\prime}} \varepsilon \delta \kappa^{\prime}  \bigg |_{0}^{S_0} + \partial_{\tau^{\prime}}\varepsilon\delta \tau  \bigg |_{0}^{S_0}+jump, 
\end{aligned}
\end{equation}
where $f,g$ are defined in Eq. \eqref{e26} and $jump$ is defined as
\begin{equation}
jump = (\partial_{\kappa^{\prime}} \varepsilon-(\partial_{\kappa^{\prime\prime}} \varepsilon)^{\prime})\delta \kappa \bigg |_{S_{jump}-}^{S_{jump}+} +\partial_{\kappa^{\prime\prime}} \varepsilon \delta \kappa^{\prime}  \bigg |_{S_{jump}-}^{S_{jump}+}  + \partial_{\tau^{\prime}}\varepsilon\delta \tau  \bigg |_{S_{jump}-}^{S_{jump}+}.
\end{equation}

\boldmath{$t(0)\geq \cot{\theta^*}$}.
\unboldmath
In this situation, the system is covered with crease-crease domains, crease-edge domains and possible plane domains. the elastic energy of the system is
\begin{equation}\label{a4}
   E=\int_0^{S_c^{*}} \varepsilon^*(S)\d S + \int_{S_2^{*}}^{S_0} \varepsilon(S)\d S,
\end{equation}
where $\varepsilon$ is defined in Eq. \eqref{e23} and $\varepsilon^*$ is defined in Eq. \eqref{evarestar}. While $S_2$ is a moving boundary, the energy variation is
\begin{equation}\label{a5}
\begin{aligned}
\delta E&=\int_{0}^{S_c^*} f^*\delta \kappa \d S+ \int_{0}^{S_c^*} g^* \delta \tau \d S 
 + \int_{S_2^*}^{S_0} f\delta \kappa \d S+ \int_{S_2^*}^{S_0} g\delta \tau \d S 
 \\&+(\partial_{\kappa^{\prime}} \varepsilon^*-(\partial_{\kappa^{\prime\prime}} \varepsilon^*)^{\prime})\delta \kappa \bigg |_{0}^{S_c^*} +\partial_{\kappa^{\prime\prime}} \varepsilon^* \delta \kappa^{\prime}  \bigg |_{0}^{S_c^*} + \partial_{\tau^{\prime}}\varepsilon^*\delta \tau  \bigg |_{0}^{S_c^*}
 \\&+(\partial_{\kappa^{\prime}} \varepsilon-(\partial_{\kappa^{\prime\prime}} \varepsilon)^{\prime})\delta \kappa \bigg |_{S_2^*+}^{S_0} +\partial_{\kappa^{\prime\prime}} \varepsilon \delta \kappa^{\prime}  \bigg |_{S_2^*+}^{S_0} + \partial_{\tau^{\prime}}\varepsilon\delta \tau  \bigg |_{S_2^*+}^{S_0}-\varepsilon(S_2^*+) \delta S_2^*+jump, 
\end{aligned}
\end{equation}
where $f,g$ are defined in Eq. \eqref{e26} and $f^*,g^*$ are defined in Eq. \eqref{eccd}. 

As analyzed in Section \ref{class}, $\kappa_n(S_c^*)=0,\kappa_n^{\prime}(S_c^*)=0$ and $\tau(S_c^*)=0$, leading to $\delta \kappa=0,\delta\kappa^{\prime}=0$ and $\delta\tau=0$ at $S_c^*$. Since a generator connects points $S=0$ and $S=S_2^*$, geometric correlations exist that $S_2^*$ is the function of $\kappa(0),\kappa^{\prime}(0), \tau(0)$ and $\kappa(S_2^*-)$ is the function of $\kappa(0),\kappa^{\prime}(0),\kappa^{\prime\prime}(0),\tau(0),\tau^{\prime}(0)$. To avoid energy singularity, $\kappa(S_2^*+)$ couples with $\kappa(S_2^*-)$ to keep the continuity of $\kappa/\kappa_g$. $\kappa(S_2^*+)$ is thus expressed as
\begin{equation}\label{A6}
\begin{aligned}
        &S_2^*=S_2^*(\kappa(0),\kappa^{\prime}(0),\tau(0)),\\
        &\kappa(S_2^*+)=\kappa(\kappa(0),\kappa^{\prime}(0),\kappa^{\prime\prime}(0),\tau(0),\tau^{\prime}(0)),
\end{aligned}
\end{equation}
leading to the variation
\begin{equation}
\begin{aligned}
    &\delta S_2^*=\frac{\partial S_2^*}{\partial \kappa(0)}\delta \kappa(0)+\frac{\partial S_2^*}{\partial \kappa^{\prime}(0)}\delta \kappa^{\prime}(0)+\frac{\partial S_2^*}{\partial \tau(0)}\delta \tau(0),\\
    &\delta\kappa(S_2^*+)=\frac{\partial \kappa(S_2^*+)}{\partial \kappa(0)}\delta \kappa(0)+\frac{\partial \kappa(S_2^*+)}{\partial \kappa^{\prime}(0)}\delta \kappa^{\prime}(0)+\frac{\partial \kappa(S_2^*+)}{\partial \kappa^{\prime\prime}(0)}\delta \kappa^{\prime\prime}(0)+\frac{\partial \kappa(S_2^*+)}{\partial \tau(0)}\delta \tau(0)+\frac{\partial \kappa(S_2^*+)}{\partial \tau^{\prime}(0)}\delta \tau^{\prime}(0).
\end{aligned}
\end{equation}
Since the plane domain may occur at $S=S_2^*$ leading to discontinuous $\tau/\kappa_g$ and $\kappa^{\prime}/\kappa_g$, variations $\delta \kappa^{\prime}(S_2^*+)$ and $\delta \tau(S_2^*+)$ are free and irrelevant with $\delta \kappa^{\prime}(S_2^*-)$ and $\delta \tau(S_2^*-)$. The energy variation in Eq. \eqref{a5} is then converted into
\begin{equation} \label{A8}
    \begin{aligned}
        \delta E&=\int_{0}^{S_c^*} f^*\delta \kappa \d S+ \int_{0}^{S_c^*} g^* \delta \tau \d S 
 + \int_{S_2^*}^{S_0} f\delta \kappa \d S+ \int_{S_2^*}^{S_0} g\delta \tau \d S 
 \\&-[(\partial_{\kappa^{\prime}} \varepsilon^*(0)-(\partial_{\kappa^{\prime\prime}} \varepsilon^*(0))^{\prime})+(\partial_{\kappa^{\prime}} \varepsilon(S_2^*+)-(\partial_{\kappa^{\prime\prime}} \varepsilon(S_2^*+))^{\prime})\frac{\partial \kappa(S_2^*+)}{\partial \kappa(0)}-\varepsilon(S_2^*+)\frac{\partial S_2^*}{\partial \kappa(0)}] \delta \kappa(0)\\
 &-[\partial_{\kappa^{\prime\prime}} \varepsilon^*(0)+(\partial_{\kappa^{\prime}} \varepsilon(S_2^*+)-(\partial_{\kappa^{\prime\prime}} \varepsilon(S_2^*+))^{\prime})\frac{\partial \kappa(S_2^*+)}{\partial \kappa^{\prime}(0)}-\varepsilon(S_2^*+)\frac{\partial S_2^*}{\partial \kappa^{\prime}(0)}]\delta\kappa^{\prime}(0)
-[(\partial_{\kappa^{\prime}} \varepsilon(S_2^*+)-(\partial_{\kappa^{\prime\prime}} \varepsilon(S_2^*+))^{\prime})\frac{\partial \kappa(S_2^*+)}{\partial \kappa^{\prime\prime}(0)}]\delta\kappa^{\prime\prime}(0)\\
&-[\partial_{\tau^{\prime}}\varepsilon^*(0)+(\partial_{\kappa^{\prime}} \varepsilon(S^*_2+)-(\partial_{\kappa^{\prime\prime}} \varepsilon(S^*_2+))\frac{\partial \kappa(S_2^*+)}{\partial \tau(0)}-\varepsilon(S_2^*+)\frac{\partial S_2^*}{\partial \tau(0)}]\delta \tau(0)
-[(\partial_{\kappa^{\prime}} \varepsilon(S_2^*+)-(\partial_{\kappa^{\prime\prime}} \varepsilon(S_2^*+))^{\prime})\frac{\partial \kappa(S_2^*+)}{\partial \tau^{\prime}(0)}]\delta \tau^{\prime}(0)\\
&+(\partial_{\kappa^{\prime}} \varepsilon(S_0)-(\partial_{\kappa^{\prime\prime}} \varepsilon(S_0))^{\prime})\delta \kappa(S_0) +\partial_{\kappa^{\prime\prime}} \varepsilon \delta \kappa^{\prime}  \bigg |_{S_2^*+}^{S_0} + \partial_{\tau^{\prime}}\varepsilon\delta \tau  \bigg |_{S_2^*+}^{S_0}+jump.
    \end{aligned}
\end{equation}
\subsection{Equilibrium equations and boundary conditions}\label{A2}
From Eq. \eqref{a5} and \eqref{A8}, we derive equilibrium equations and boundary conditions for both situations based on energy minimization $\delta E=0$.

\boldmath{$t(0)<\cot{\theta^*}$}. 
\unboldmath
For freely-deformed systems, $\delta \kappa$ and $\delta \tau$ are arbitrary. $\delta E=0$ thus leads to $f=0$ and $g=0$ satisfied along the crease $(0,S_0)$, and the boundary conditions $\partial_{\kappa^{\prime}} \varepsilon-(\partial_{\kappa^{\prime\prime}} \varepsilon)^{\prime}=0$, $\partial_{\kappa^{\prime\prime}} \varepsilon=0$, $\partial_{\tau^{\prime}}\varepsilon=0$ on both ends $S=0, S=S_0$.

$jump=0$ introduces extra jump conditions at $S=S_{jump}$. Due to the continuity of $\kappa/\kappa_g$, $\delta\kappa(S_{jump}+)$ couples with $\delta\kappa(S_{jump}-)$, following $\delta\kappa(S_{jump}+)/\delta\kappa(S_{jump}-)=\kappa_g(S_{jump}+)/\kappa_g(S_{jump}-)$. In the contrast, $\delta\kappa^{\prime}(S_{jump}+)$ and $\delta \tau(S_{jump}+)$ are irrelevant with $\delta\kappa^{\prime}(S_{jump}-)$ and $\delta \tau(S_{jump}-)$. $jump=0$ thus leads to
\begin{equation} \label{Aejump}
    \begin{aligned}
        &[\partial_{\kappa^{\prime}} \varepsilon(S_{jump}-)-(\partial_{\kappa^{\prime\prime}} \varepsilon(S_{jump}-))^{\prime}]\kappa_g(S_{jump}-)
        -[\partial_{\kappa^{\prime}} \varepsilon(S_{jump}+)-(\partial_{\kappa^{\prime\prime}} \varepsilon(S_{jump}+))^{\prime}]\kappa_g(S_{jump}+)=0,\\&\partial_{\kappa^{\prime\prime}}\varepsilon(S_{jump}+)=\partial_{\kappa^{\prime\prime}}\varepsilon(S_{jump}-)=0, \\
        &\partial_{\tau^{\prime}}\varepsilon(S_{jump}+)=\partial_{\tau^{\prime}}\varepsilon(S_{jump}-)=0. 
    \end{aligned}
\end{equation}
Satisfying the jump conditions in Eq. \eqref{Aejump} may lead to jumps in $\tau$ and $\kappa^{\prime}$ at $S=S_{jumps}$, forming a plane domain.

\boldmath{$t(0)\geq\cot{\theta^*}$}. 
\unboldmath
For freely-deformed systems, $\delta \kappa$ and $\delta \tau$ are arbitrary. $\delta E=0$ thus leads to $f^*=0,g^*=0$ satisfied along $(0,S_c^*)$ and $f=0,g=0$ satisfied along $(S_2^*,S_0)$. The boundary conditions are
\begin{equation}
\begin{aligned}
 &(\partial_{\kappa^{\prime}} \varepsilon^*(0)-(\partial_{\kappa^{\prime\prime}} \varepsilon^*(0))^{\prime})+(\partial_{\kappa^{\prime}} \varepsilon(S_2^*+)-(\partial_{\kappa^{\prime\prime}} \varepsilon(S_2^*+))^{\prime})\frac{\partial \kappa(S_2^*+)}{\partial \kappa(0)}-\varepsilon(S_2^*+)\frac{\partial S_2^*}{\partial \kappa(0)}=0\\
 &\partial_{\kappa^{\prime\prime}} \varepsilon^*(0)+(\partial_{\kappa^{\prime}} \varepsilon(S_2^*+)-(\partial_{\kappa^{\prime\prime}} \varepsilon(S_2^*+))^{\prime})\frac{\partial \kappa(S_2^*+)}{\partial \kappa^{\prime}(0)}-\varepsilon(S_2^*+)\frac{\partial S_2^*}{\partial \kappa^{\prime}(0)}]\delta\kappa^{\prime}(0)=0\\
&(\partial_{\kappa^{\prime}} \varepsilon(S_2^*+)-(\partial_{\kappa^{\prime\prime}} \varepsilon(S_2^*+))^{\prime})\frac{\partial \kappa(S_2^*+)}{\partial \kappa^{\prime\prime}(0)}=(\partial_{\kappa^{\prime}} \varepsilon(S_2^*+)-(\partial_{\kappa^{\prime\prime}} \varepsilon(S_2^*+))^{\prime})\frac{\partial \kappa(S_2^*+)}{\partial \tau^{\prime}(0)}=0\\
&\partial_{\tau^{\prime}}\varepsilon^*(0)+(\partial_{\kappa^{\prime}} \varepsilon(S^*_2+)-(\partial_{\kappa^{\prime\prime}} \varepsilon(S^*_2+))\frac{\partial \kappa(S_2^*+)}{\partial \tau(0)}-\varepsilon(S_2^*+)\frac{\partial S_2^*}{\partial \tau(0)}=0
\\
&\partial_{\kappa^{\prime}} \varepsilon(S_0)-(\partial_{\kappa^{\prime\prime}} \varepsilon(S_0))^{\prime}=0\\
&\partial_{\kappa^{\prime\prime}} \varepsilon \delta \kappa^{\prime}  \bigg |_{S_2^*+}^{S_0} = \partial_{\tau^{\prime}}\varepsilon\delta \tau  \bigg |_{S_2^*+}^{S_0}=0       
\end{aligned}
\end{equation}
Besides, the distributions of $\kappa,\tau$ should satisfy jump conditions in Eq. \eqref{Aejump} at $S=S_{jump}$ and geometric constraints in Eq. \eqref{A6}.
If the solutions exist for both situations, the system chooses the distribution of $\kappa$ and $\tau$ with the less elastic energy expressed in Eq. \eqref{a1} and \eqref{a4}. 

Besides the distributions illustrated in Fig. \ref{fig:A1}, there are some special cases shown in Fig. \ref{fig:A2}. If a plane domain shown in Fig. \ref{fig:A2}(a) is formed, the crease-crease domain may not start from $S=0$ or $S=S_0$. Thus, the elastic energy is
\begin{equation} \label{a11}
    E=\int_0^{S_1^{*}} \varepsilon(S)\d S + \int_{S_1^{*}}^{S_c^*} \varepsilon^*(S)\d S+\int_{S_2^{*}}^{S_0} \varepsilon(S)\d S,
\end{equation}
where $S_1^*$ and $S_2^*$ are moving boundaries. If $\kappa_g$ changes its sign along the fold, the system may have more crease-crease domains, as shown in Fig. \ref{fig:A2}(b). The elastic energy is
\begin{equation}\label{a12}
    E=\int_0^{S_c^{*}} \varepsilon^*(S)\d S + \int_{S_2^{*}}^{S_1^{**}} \varepsilon(S)\d S+\int_{S_1^{**}}^{S_c^{**}} \varepsilon^*(S)\d S,
\end{equation}
where $S_2^*$ and $S_1^{**}$ are moving boundaries. From \eqref{a11} and \eqref{a12}, we can derive the equilibrium equations and boundary conditions for both cases, following the variational method given in this appendix. 
\begin{figure}[t!]
    \centering
    \includegraphics[width=1\textwidth]{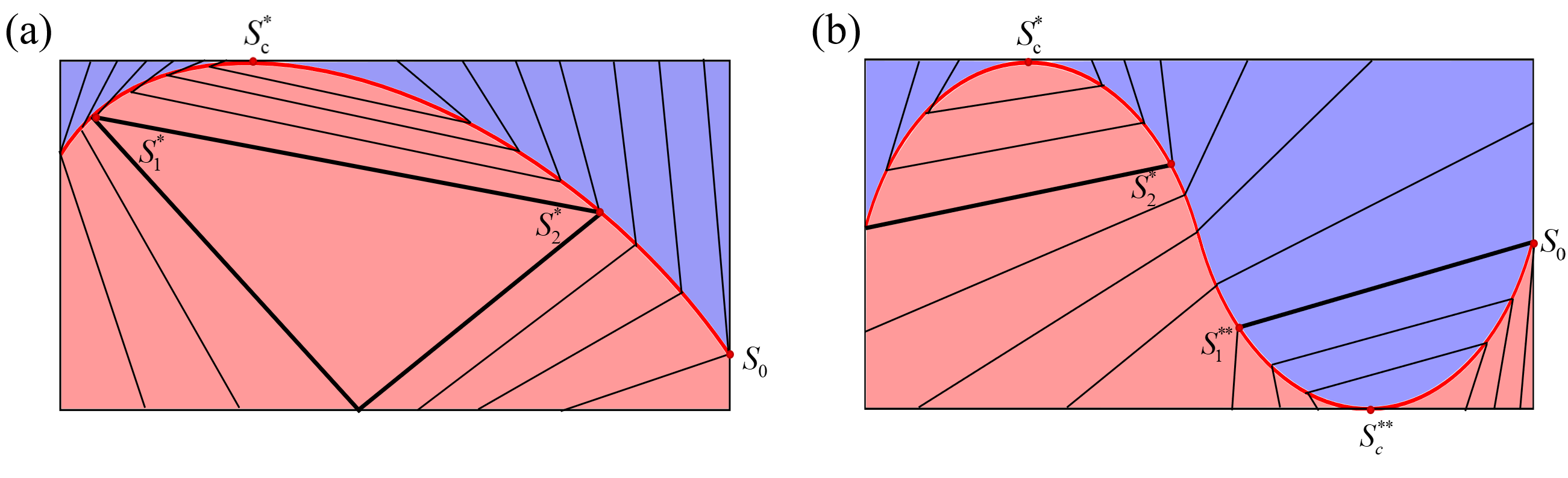}
    \caption{\label{fig:A2}Two special cases of generator distributions in a single-curve-fold system (panel $+$: red, panel $-$: blue). (a) System with a plane domain formed between $(S_1^*,S_2^*)$. Thick generators represent the boundaries of the plane domain. (b) System with two crease-crease domains. $\kappa_g$ changes its sign along the fold. Thick generators represent boundaries of the crease-crease domains.}
\end{figure}

For systems with geometric constraints, the equilibrium equations and boundary conditions can be obtained by converting Eq. \eqref{e29} into Eq. \eqref{a5} and \eqref{A8}. Following Eq. \eqref{e31} - \eqref{e37}, we can derive equilibrium equations satisfied in different domains (which we have studied in Section \ref{sec2333}) and the boundary conditions. 
\bibliographystyle{elsarticle-harv} 
\bibliography{cas-refs}

\begin{thebibliography}{50}
\expandafter\ifx\csname natexlab\endcsname\relax\def\natexlab#1{#1}\fi
\providecommand{\url}[1]{\texttt{#1}}
\providecommand{\href}[2]{#2}
\providecommand{\path}[1]{#1}
\providecommand{\DOIprefix}{doi:}
\providecommand{\ArXivprefix}{arXiv:}
\providecommand{\URLprefix}{URL: }
\providecommand{\Pubmedprefix}{pmid:}
\providecommand{\doi}[1]{\href{http://dx.doi.org/#1}{\path{#1}}}
\providecommand{\Pubmed}[1]{\href{pmid:#1}{\path{#1}}}
\providecommand{\bibinfo}[2]{#2}
\ifx\xfnm\relax \def\xfnm[#1]{\unskip,\space#1}\fi
\bibitem[{Audoly and Van~der Heijden(2023)}]{audoly2023analysis}
\bibinfo{author}{Audoly, B.}, \bibinfo{author}{Van~der Heijden, G.},
  \bibinfo{year}{2023}.
\newblock \bibinfo{title}{Analysis of cone-like singularities in twisted
  elastic ribbons}.
\newblock \bibinfo{journal}{Journal of the Mechanics and Physics of Solids}
  \bibinfo{volume}{171}, \bibinfo{pages}{105131}.
\bibitem[{Baek et~al.(2020)Baek, Yim, Chae, Lee and
  Cho}]{doi:10.1126/scirobotics.aaz6262}
\bibinfo{author}{Baek, S.M.}, \bibinfo{author}{Yim, S.}, \bibinfo{author}{Chae,
  S.H.}, \bibinfo{author}{Lee, D.Y.}, \bibinfo{author}{Cho, K.J.},
  \bibinfo{year}{2020}.
\newblock \bibinfo{title}{Ladybird beetle inspired compliant origami}.
\newblock \bibinfo{journal}{Science Robotics} \bibinfo{volume}{5},
  \bibinfo{pages}{eaaz6262}.
\bibitem[{Bobenko and Suris(2008)}]{bobenko2008discrete}
\bibinfo{author}{Bobenko, A.I.}, \bibinfo{author}{Suris, Y.B.},
  \bibinfo{year}{2008}.
\newblock \bibinfo{title}{Discrete differential geometry: integrable
  structure}. volume~\bibinfo{volume}{98}.
\newblock \bibinfo{publisher}{American Mathematical Soc.}
\bibitem[{Chen et~al.(2022)Chen, Fosdick and Fried}]{CHEN2022105068}
\bibinfo{author}{Chen, Y.C.}, \bibinfo{author}{Fosdick, R.},
  \bibinfo{author}{Fried, E.}, \bibinfo{year}{2022}.
\newblock \bibinfo{title}{A novel dimensional reduction for the equilibrium
  study of inextensional material surfaces}.
\newblock \bibinfo{journal}{Journal of the Mechanics and Physics of Solids}
  \bibinfo{volume}{169}, \bibinfo{pages}{105068}.
\bibitem[{Cheng et~al.(2023)Cheng, Fan, Yao, Jin, Lv, Lan, Bo, Chen, Zhang,
  Shen et~al.}]{cheng2023programming}
\bibinfo{author}{Cheng, X.}, \bibinfo{author}{Fan, Z.}, \bibinfo{author}{Yao,
  S.}, \bibinfo{author}{Jin, T.}, \bibinfo{author}{Lv, Z.},
  \bibinfo{author}{Lan, Y.}, \bibinfo{author}{Bo, R.}, \bibinfo{author}{Chen,
  Y.}, \bibinfo{author}{Zhang, F.}, \bibinfo{author}{Shen, Z.}, et~al.,
  \bibinfo{year}{2023}.
\newblock \bibinfo{title}{Programming 3d curved mesosurfaces using microlattice
  designs}.
\newblock \bibinfo{journal}{Science} \bibinfo{volume}{379},
  \bibinfo{pages}{1225--1232}.
\bibitem[{Demaine et~al.()Demaine, Demaine, Huffman, Koschitz and
  Tachi}]{demaine2015characterization}
\bibinfo{author}{Demaine, E.D.}, \bibinfo{author}{Demaine, M.L.},
  \bibinfo{author}{Huffman, D.A.}, \bibinfo{author}{Koschitz, D.},
  \bibinfo{author}{Tachi, T.}, .
\newblock \bibinfo{title}{Characterization of curved creases and rulings:
  Design and analysis of lens tessellations}.
\newblock \bibinfo{journal}{arXiv:1502.03191} .
\bibitem[{Demaine et~al.(2011)Demaine, Demaine, Koschitz and
  Tachi}]{demaine2011curved}
\bibinfo{author}{Demaine, E.D.}, \bibinfo{author}{Demaine, M.L.},
  \bibinfo{author}{Koschitz, D.}, \bibinfo{author}{Tachi, T.},
  \bibinfo{year}{2011}.
\newblock \bibinfo{title}{Curved crease folding: a review on art, design and
  mathematics}, in: \bibinfo{booktitle}{Proceedings of the IABSE-IASS
  symposium: taller, longer, lighter}, \bibinfo{organization}{Citeseer}. pp.
  \bibinfo{pages}{20--23}.
\bibitem[{Dias(2012)}]{dias2012thesis}
\bibinfo{author}{Dias, M.A.}, \bibinfo{year}{2012}.
\newblock \bibinfo{title}{Swelling and folding as mechanisms of 3D shape
  formation in thin elastic sheets}.
\newblock Ph.D. thesis. University of Massachusetts Amherst.
\bibitem[{Dias and Audoly(2014)}]{DIAS201457}
\bibinfo{author}{Dias, M.A.}, \bibinfo{author}{Audoly, B.},
  \bibinfo{year}{2014}.
\newblock \bibinfo{title}{A non-linear rod model for folded elastic strips}.
\newblock \bibinfo{journal}{Journal of the Mechanics and Physics of Solids}
  \bibinfo{volume}{62}, \bibinfo{pages}{57--80}.
\newblock \bibinfo{note}{Sixtieth anniversary issue in honor of Professor
  Rodney Hill}.
\bibitem[{Dias and Audoly(2015)}]{dias2015wunderlich}
\bibinfo{author}{Dias, M.A.}, \bibinfo{author}{Audoly, B.},
  \bibinfo{year}{2015}.
\newblock \bibinfo{title}{“{W}underlich, meet {K}irchhoff”: A general and
  unified description of elastic ribbons and thin rods}.
\newblock \bibinfo{journal}{Journal of Elasticity} \bibinfo{volume}{119},
  \bibinfo{pages}{49--66}.
\bibitem[{Dias et~al.(2012)Dias, Dudte, Mahadevan and
  Santangelo}]{dias2012geometric}
\bibinfo{author}{Dias, M.A.}, \bibinfo{author}{Dudte, L.H.},
  \bibinfo{author}{Mahadevan, L.}, \bibinfo{author}{Santangelo, C.D.},
  \bibinfo{year}{2012}.
\newblock \bibinfo{title}{Geometric mechanics of curved crease origami}.
\newblock \bibinfo{journal}{Physical review letters} \bibinfo{volume}{109},
  \bibinfo{pages}{114301}.
\bibitem[{Dias and Santangelo(2012)}]{Dias_2012}
\bibinfo{author}{Dias, M.A.}, \bibinfo{author}{Santangelo, C.D.},
  \bibinfo{year}{2012}.
\newblock \bibinfo{title}{The shape and mechanics of curved-fold origami
  structures}.
\newblock \bibinfo{journal}{Europhysics Letters} \bibinfo{volume}{100},
  \bibinfo{pages}{54005}.
\bibitem[{Do~Carmo(1976)}]{do1976differential}
\bibinfo{author}{Do~Carmo, M.P.}, \bibinfo{year}{1976}.
\newblock \bibinfo{title}{Differential geometry of curves and surfaces}.
\newblock \bibinfo{journal}{Englewood Cliffs, New Jersey} .
\bibitem[{Duffy et~al.(2021)Duffy, Cmok, Biggins, Krishna, Modes, Abdelrahman,
  Javed, Ware, Feng and Warner}]{Fannon}
\bibinfo{author}{Duffy, D.}, \bibinfo{author}{Cmok, L.},
  \bibinfo{author}{Biggins, J.S.}, \bibinfo{author}{Krishna, A.},
  \bibinfo{author}{Modes, C.D.}, \bibinfo{author}{Abdelrahman, M.K.},
  \bibinfo{author}{Javed, M.}, \bibinfo{author}{Ware, T.H.},
  \bibinfo{author}{Feng, F.}, \bibinfo{author}{Warner, M.},
  \bibinfo{year}{2021}.
\newblock \bibinfo{title}{{Shape programming lines of concentrated Gaussian
  curvature}}.
\newblock \bibinfo{journal}{Journal of Applied Physics} \bibinfo{volume}{129},
  \bibinfo{pages}{224701}.
\bibitem[{Duncan and Duncan(1982)}]{duncan1982folded}
\bibinfo{author}{Duncan, J.P.}, \bibinfo{author}{Duncan, J.},
  \bibinfo{year}{1982}.
\newblock \bibinfo{title}{Folded developables}.
\newblock \bibinfo{journal}{Proceedings of the Royal Society of London. A.
  Mathematical and Physical Sciences} \bibinfo{volume}{383},
  \bibinfo{pages}{191--205}.
\bibitem[{Evans et~al.(2015)Evans, Silverberg and
  Santangelo}]{evans2015lattice}
\bibinfo{author}{Evans, A.A.}, \bibinfo{author}{Silverberg, J.L.},
  \bibinfo{author}{Santangelo, C.D.}, \bibinfo{year}{2015}.
\newblock \bibinfo{title}{Lattice mechanics of origami tessellations}.
\newblock \bibinfo{journal}{Physical Review E} \bibinfo{volume}{92},
  \bibinfo{pages}{013205}.
\bibitem[{Feng et~al.(2024)Feng, Dradrach, Zmyślony, Barnes and
  Biggins}]{feng2024geometry}
\bibinfo{author}{Feng, F.}, \bibinfo{author}{Dradrach, K.},
  \bibinfo{author}{Zmyślony, M.}, \bibinfo{author}{Barnes, M.},
  \bibinfo{author}{Biggins, J.S.}, \bibinfo{year}{2024}.
\newblock \bibinfo{title}{Geometry{,} mechanics and actuation of intrinsically
  curved folds}.
\newblock \bibinfo{journal}{Soft Matter} \bibinfo{volume}{20},
  \bibinfo{pages}{2132--2140}.
\bibitem[{Feng et~al.(2022)Feng, Duffy, Warner and
  Biggins}]{feng2022interfacial}
\bibinfo{author}{Feng, F.}, \bibinfo{author}{Duffy, D.},
  \bibinfo{author}{Warner, M.}, \bibinfo{author}{Biggins, J.S.},
  \bibinfo{year}{2022}.
\newblock \bibinfo{title}{Interfacial metric mechanics: stitching patterns of
  shape change in active sheets}.
\newblock \bibinfo{journal}{Proceedings of the Royal Society A}
  \bibinfo{volume}{478}, \bibinfo{pages}{20220230}.
\bibitem[{Flores et~al.(2022)Flores, Stein-Montalvo and Adriaenssens}]{water}
\bibinfo{author}{Flores, J.}, \bibinfo{author}{Stein-Montalvo, L.},
  \bibinfo{author}{Adriaenssens, S.}, \bibinfo{year}{2022}.
\newblock \bibinfo{title}{Effect of crease curvature on the bistability of the
  origami waterbomb base}.
\newblock \bibinfo{journal}{Extreme Mechanics Letters} \bibinfo{volume}{57},
  \bibinfo{pages}{101909}.
\bibitem[{Fuchs and Tabachnikov(1999)}]{10.2307/2589583}
\bibinfo{author}{Fuchs, D.}, \bibinfo{author}{Tabachnikov, S.},
  \bibinfo{year}{1999}.
\newblock \bibinfo{title}{More on paperfolding}.
\newblock \bibinfo{journal}{The American Mathematical Monthly}
  \bibinfo{volume}{106}, \bibinfo{pages}{27--35}.
\bibitem[{Huffman(1976)}]{1674542}
\bibinfo{author}{Huffman}, \bibinfo{year}{1976}.
\newblock \bibinfo{title}{Curvature and creases: A primer on paper}.
\newblock \bibinfo{journal}{IEEE Transactions on Computers}
  \bibinfo{volume}{C-25}, \bibinfo{pages}{1010--1019}.
\bibitem[{Jiang et~al.(2019)Jiang, Mundilova, Rist, Wallner and
  Pottmann}]{jiang2019curve}
\bibinfo{author}{Jiang, C.}, \bibinfo{author}{Mundilova, K.},
  \bibinfo{author}{Rist, F.}, \bibinfo{author}{Wallner, J.},
  \bibinfo{author}{Pottmann, H.}, \bibinfo{year}{2019}.
\newblock \bibinfo{title}{Curve-pleated structures}.
\newblock \bibinfo{journal}{ACM Transactions on Graphics (TOG)}
  \bibinfo{volume}{38}, \bibinfo{pages}{1--13}.
\bibitem[{Karami et~al.(2024)Karami, Reddy and Nassar}]{PRLmeta}
\bibinfo{author}{Karami, A.}, \bibinfo{author}{Reddy, A.},
  \bibinfo{author}{Nassar, H.}, \bibinfo{year}{2024}.
\newblock \bibinfo{title}{Curved-crease origami for morphing metamaterials}.
\newblock \bibinfo{journal}{Phys. Rev. Lett.} \bibinfo{volume}{132},
  \bibinfo{pages}{108201}.
\bibitem[{Kilian et~al.(2008)Kilian, Fl{\"o}ry, Chen, Mitra, Sheffer and
  Pottmann}]{kilian2008curved}
\bibinfo{author}{Kilian, M.}, \bibinfo{author}{Fl{\"o}ry, S.},
  \bibinfo{author}{Chen, Z.}, \bibinfo{author}{Mitra, N.J.},
  \bibinfo{author}{Sheffer, A.}, \bibinfo{author}{Pottmann, H.},
  \bibinfo{year}{2008}.
\newblock \bibinfo{title}{Curved folding}.
\newblock \bibinfo{journal}{ACM transactions on graphics (TOG)}
  \bibinfo{volume}{27}, \bibinfo{pages}{1--9}.
\bibitem[{Koschitz et~al.(2008)Koschitz, Demaine and
  Demaine}]{CurvedCrease_AAG2008}
\bibinfo{author}{Koschitz, D.}, \bibinfo{author}{Demaine, E.D.},
  \bibinfo{author}{Demaine, M.L.}, \bibinfo{year}{2008}.
\newblock \bibinfo{title}{Curved crease origami}, in:
  \bibinfo{booktitle}{Abstracts from Advances in Architectural Geometry (AAG
  2008)}, \bibinfo{address}{Vienna, Austria}. pp. \bibinfo{pages}{29--32}.
\bibitem[{Lee et~al.(2021)Lee, Chen, Heitzmann and Gattas}]{lee2021compliant}
\bibinfo{author}{Lee, T.U.}, \bibinfo{author}{Chen, Y.},
  \bibinfo{author}{Heitzmann, M.T.}, \bibinfo{author}{Gattas, J.M.},
  \bibinfo{year}{2021}.
\newblock \bibinfo{title}{Compliant curved-crease origami-inspired
  metamaterials with a programmable force-displacement response}.
\newblock \bibinfo{journal}{Materials \& Design} \bibinfo{volume}{207},
  \bibinfo{pages}{109859}.
\bibitem[{Lee et~al.(2018)Lee, You and Gattas}]{lee2018elastica}
\bibinfo{author}{Lee, T.U.}, \bibinfo{author}{You, Z.},
  \bibinfo{author}{Gattas, J.M.}, \bibinfo{year}{2018}.
\newblock \bibinfo{title}{Elastica surface generation of curved-crease
  origami}.
\newblock \bibinfo{journal}{International Journal of Solids and Structures}
  \bibinfo{volume}{136}, \bibinfo{pages}{13--27}.
\bibitem[{Liu and James(2024)}]{liu_design_2024}
\bibinfo{author}{Liu, H.}, \bibinfo{author}{James, R.D.}, \bibinfo{year}{2024}.
\newblock \bibinfo{title}{Design of origami structures with curved tiles
  between the creases}.
\newblock \bibinfo{journal}{Journal of the Mechanics and Physics of Solids}
  \bibinfo{volume}{185}, \bibinfo{pages}{105559}.
\bibitem[{Mitani(2011)}]{mitani2011design}
\bibinfo{author}{Mitani, J.}, \bibinfo{year}{2011}.
\newblock \bibinfo{title}{A design method for axisymmetric curved origami with
  triangular prism protrusions}, in: \bibinfo{booktitle}{Proceedings of 5th
  International Conference on Origami in Science, Mathematics, and Education
  (5OSME)}.
\bibitem[{Mitani(2019)}]{mitani2019curved}
\bibinfo{author}{Mitani, J.}, \bibinfo{year}{2019}.
\newblock \bibinfo{title}{Curved-folding origami design}.
\newblock \bibinfo{publisher}{CRC Press}.
\bibitem[{Mitani and Igarashi(2011)}]{mitani2011interactive}
\bibinfo{author}{Mitani, J.}, \bibinfo{author}{Igarashi, T.},
  \bibinfo{year}{2011}.
\newblock \bibinfo{title}{{Interactive Design of Planar Curved Folding by
  Reflection}}, in: \bibinfo{editor}{Chen, B.Y.}, \bibinfo{editor}{Kautz, J.},
  \bibinfo{editor}{Lee, T.Y.}, \bibinfo{editor}{Lin, M.C.} (Eds.),
  \bibinfo{booktitle}{Pacific Graphics Short Papers}, \bibinfo{publisher}{The
  Eurographics Association}.
\bibitem[{Moulton et~al.(2013)Moulton, Lessinnes and Goriely}]{MOULTON2013398}
\bibinfo{author}{Moulton, D.}, \bibinfo{author}{Lessinnes, T.},
  \bibinfo{author}{Goriely, A.}, \bibinfo{year}{2013}.
\newblock \bibinfo{title}{Morphoelastic rods. part i: A single growing elastic
  rod}.
\newblock \bibinfo{journal}{Journal of the Mechanics and Physics of Solids}
  \bibinfo{volume}{61}, \bibinfo{pages}{398--427}.
\bibitem[{Mouthuy et~al.(2012)Mouthuy, Coulombier, Pardoen, Raskin and
  Jonas}]{mouthuy2012overcurvature}
\bibinfo{author}{Mouthuy, P.O.}, \bibinfo{author}{Coulombier, M.},
  \bibinfo{author}{Pardoen, T.}, \bibinfo{author}{Raskin, J.P.},
  \bibinfo{author}{Jonas, A.M.}, \bibinfo{year}{2012}.
\newblock \bibinfo{title}{Overcurvature describes the buckling and folding of
  rings from curved origami to foldable tents}.
\newblock \bibinfo{journal}{Nature communications} \bibinfo{volume}{3},
  \bibinfo{pages}{1290}.
\bibitem[{M{\"u}ller and Vaxman(2021)}]{Muller2021}
\bibinfo{author}{M{\"u}ller, C.}, \bibinfo{author}{Vaxman, A.},
  \bibinfo{year}{2021}.
\newblock \bibinfo{title}{Discrete curvature and torsion from cross-ratios}.
\newblock \bibinfo{journal}{Annali di Matematica Pura ed Applicata (1923 -)}
  \bibinfo{volume}{200}, \bibinfo{pages}{1935--1960}.
\bibitem[{Mundilova et~al.(2023)Mundilova, Demaine, Lang and
  Tachi}]{mundilovacurved}
\bibinfo{author}{Mundilova, K.}, \bibinfo{author}{Demaine, E.D.},
  \bibinfo{author}{Lang, R.}, \bibinfo{author}{Tachi, T.},
  \bibinfo{year}{2023}.
\newblock \bibinfo{title}{Curved-crease origami spirals constructed from
  reflected cones}, in: \bibinfo{editor}{Holdener, J.},
  \bibinfo{editor}{Torrence, E.}, \bibinfo{editor}{Fong, C.},
  \bibinfo{editor}{Seaton, K.} (Eds.), \bibinfo{booktitle}{Proceedings of
  Bridges 2023: Mathematics, Art, Music, Architecture, Culture},
  \bibinfo{publisher}{Tessellations Publishing}, \bibinfo{address}{Phoenix,
  Arizona}. pp. \bibinfo{pages}{401--404}.
\bibitem[{Rabinovich et~al.(2018)Rabinovich, Hoffmann and
  Sorkine-Hornung}]{10.1145/3180494}
\bibinfo{author}{Rabinovich, M.}, \bibinfo{author}{Hoffmann, T.},
  \bibinfo{author}{Sorkine-Hornung, O.}, \bibinfo{year}{2018}.
\newblock \bibinfo{title}{Discrete geodesic nets for modeling developable
  surfaces}.
\newblock \bibinfo{journal}{ACM Trans. Graph.} \bibinfo{volume}{37}.
\bibitem[{Rus and Sung(2018)}]{doi:10.1126/scirobotics.aat0938}
\bibinfo{author}{Rus, D.}, \bibinfo{author}{Sung, C.}, \bibinfo{year}{2018}.
\newblock \bibinfo{title}{Spotlight on origami robots}.
\newblock \bibinfo{journal}{Science Robotics} \bibinfo{volume}{3},
  \bibinfo{pages}{eaat0938}.
\bibitem[{Sadowsky(1930)}]{sadowsky1930theorie}
\bibinfo{author}{Sadowsky, M.}, \bibinfo{year}{1930}.
\newblock \bibinfo{title}{Theorie der elastisch biegsamen undehnbaren
  b{\"a}nder mit anwendungen auf das m{\"o}bius’ sche band}.
\newblock \bibinfo{journal}{Verhandl. des} \bibinfo{volume}{3},
  \bibinfo{pages}{444--451}.
\bibitem[{Sasaki and Mitani(2022)}]{sasaki2022simple}
\bibinfo{author}{Sasaki, K.}, \bibinfo{author}{Mitani, J.},
  \bibinfo{year}{2022}.
\newblock \bibinfo{title}{Simple implementation and low computational cost
  simulation of curved folds based on ruling-aware triangulation}.
\newblock \bibinfo{journal}{Computers \& Graphics} \bibinfo{volume}{102},
  \bibinfo{pages}{213--219}.
\bibitem[{Solomon et~al.(2012)Solomon, Vouga, Wardetzky and Grinspun}]{solomon}
\bibinfo{author}{Solomon, J.}, \bibinfo{author}{Vouga, E.},
  \bibinfo{author}{Wardetzky, M.}, \bibinfo{author}{Grinspun, E.},
  \bibinfo{year}{2012}.
\newblock \bibinfo{title}{Flexible developable surfaces}.
\newblock \bibinfo{journal}{Computer Graphics Forum} \bibinfo{volume}{31},
  \bibinfo{pages}{1567--1576}.
\bibitem[{Starostin and van~der Heijden(2007)}]{starostin2007shape}
\bibinfo{author}{Starostin, E.L.}, \bibinfo{author}{van~der Heijden, G.H.},
  \bibinfo{year}{2007}.
\newblock \bibinfo{title}{The shape of a möbius strip}.
\newblock \bibinfo{journal}{Nature materials} \bibinfo{volume}{6},
  \bibinfo{pages}{563--567}.
\bibitem[{Sun et~al.(2024)Sun, Song, Ju and Zhou}]{meta3}
\bibinfo{author}{Sun, Y.}, \bibinfo{author}{Song, K.}, \bibinfo{author}{Ju,
  J.}, \bibinfo{author}{Zhou, X.}, \bibinfo{year}{2024}.
\newblock \bibinfo{title}{Curved-creased origami mechanical metamaterials with
  programmable stabilities and stiffnesses}.
\newblock \bibinfo{journal}{International Journal of Mechanical Sciences}
  \bibinfo{volume}{262}, \bibinfo{pages}{108729}.
\bibitem[{Tachi(2013)}]{tachi2013composite}
\bibinfo{author}{Tachi, T.}, \bibinfo{year}{2013}.
\newblock \bibinfo{title}{Composite rigid-foldable curved origami structure}.
\newblock \bibinfo{journal}{Proceedings of Transformables} ,
  \bibinfo{pages}{18--20}.
\bibitem[{Tachi and Epps(2011)}]{tachi2011designing}
\bibinfo{author}{Tachi, T.}, \bibinfo{author}{Epps, G.}, \bibinfo{year}{2011}.
\newblock \bibinfo{title}{Designing one-dof mechanisms for architecture by
  rationalizing curved folding}, in: \bibinfo{booktitle}{International
  Symposium on Algorithmic Design for Architecture and Urban Design
  (ALGODE-AIJ). Tokyo}, p.~\bibinfo{pages}{6}.
\bibitem[{Todres(2015)}]{todres2015translation}
\bibinfo{author}{Todres, R.E.}, \bibinfo{year}{2015}.
\newblock \bibinfo{title}{Translation of w. wunderlich’s “on a developable
  m{\"o}bius band”}.
\newblock \bibinfo{journal}{Journal of Elasticity} \bibinfo{volume}{119},
  \bibinfo{pages}{23--34}.
\bibitem[{Woodruff and Filipov(2018)}]{FEM}
\bibinfo{author}{Woodruff, S.R.}, \bibinfo{author}{Filipov, E.T.},
  \bibinfo{year}{2018}.
\newblock \bibinfo{title}{Structural Analysis of Curved Folded Deployables}.
\newblock pp. \bibinfo{pages}{793--803}.
\bibitem[{Wunderlich(1962)}]{Wunderlich1962}
\bibinfo{author}{Wunderlich, W.}, \bibinfo{year}{1962}.
\newblock \bibinfo{title}{{\"U}ber ein abwickelbares m{\"o}biusband}.
\newblock \bibinfo{journal}{Monatshefte f{\"u}r Mathematik}
  \bibinfo{volume}{66}, \bibinfo{pages}{276--289}.
\bibitem[{Yu and Hanna(2019)}]{YU2019657}
\bibinfo{author}{Yu, T.}, \bibinfo{author}{Hanna, J.}, \bibinfo{year}{2019}.
\newblock \bibinfo{title}{Bifurcations of buckled, clamped anisotropic rods and
  thin bands under lateral end translations}.
\newblock \bibinfo{journal}{Journal of the Mechanics and Physics of Solids}
  \bibinfo{volume}{122}, \bibinfo{pages}{657--685}.
\bibitem[{Zhang et~al.(2023)Zhang, Xu, Emu, Wei, Chen, Zhai, Kong, Wang and
  Jiang}]{VR}
\bibinfo{author}{Zhang, Z.}, \bibinfo{author}{Xu, Z.}, \bibinfo{author}{Emu,
  L.}, \bibinfo{author}{Wei, P.}, \bibinfo{author}{Chen, S.},
  \bibinfo{author}{Zhai, Z.}, \bibinfo{author}{Kong, L.},
  \bibinfo{author}{Wang, Y.}, \bibinfo{author}{Jiang, H.},
  \bibinfo{year}{2023}.
\newblock \bibinfo{title}{Active mechanical haptics with high-fidelity
  perceptions for immersive virtual reality}.
\newblock \bibinfo{journal}{Nature Machine Intelligence} \bibinfo{volume}{5},
  \bibinfo{pages}{643--655}.
\bibitem[{Zou et~al.(2024)Zou, Feng, Liu, Lv and Duan}]{zou2024kinematics}
\bibinfo{author}{Zou, Y.}, \bibinfo{author}{Feng, F.}, \bibinfo{author}{Liu,
  K.}, \bibinfo{author}{Lv, P.}, \bibinfo{author}{Duan, H.},
  \bibinfo{year}{2024}.
\newblock \bibinfo{title}{Kinematics and dynamics of non-developable origami}.
\newblock \bibinfo{journal}{Proceedings of the Royal Society A}
  \bibinfo{volume}{480}, \bibinfo{pages}{20230610}.

\end{thebibliography}





\end{document}